\documentclass[10pt, a4paper,reqno]{amsart}

\usepackage{color} \definecolor{bleu_sombre}{rgb}{0,0,0.6}  \definecolor{rouge_sombre}{rgb}{0.8,0,0}\definecolor{vert_sombre}{rgb}{0,0.6,0}
\usepackage[plainpages=false,colorlinks,linkcolor=bleu_sombre,citecolor=rouge_sombre,urlcolor=vert_sombre,breaklinks]{hyperref}

\usepackage[english]{babel}

\usepackage{amsmath,amssymb,amsthm,graphicx,amsfonts,url,color,enumerate,dsfont,stmaryrd}
\theoremstyle{plain}
\newtheorem{theorem}{{Theorem}}[section] %\sc{Théorème} pour avoir des petites capitales (mais ce n'est plus en gras)
\newtheorem*{theorem*}{{Theorem}}
\newtheorem{proposition}[theorem]{Proposition}
\newtheorem*{proposition*}{Proposition}
\newtheorem{corollary}[theorem]{Corollary}
\newtheorem*{corollary*}{Corollary}
\newtheorem{lemma}[theorem]{Lemma}
\newtheorem*{lemma*}{Lemma}

\theoremstyle{definition}
\newtheorem{definition}[theorem]{Definition}
\newtheorem*{definition*}{Definition}

\theoremstyle{remark}
\newtheorem{remark}[theorem]{Remark}

\makeatletter

\@addtoreset{equation}{section}  % Pour que le compteur des équations reparte à 0 en changeant de section
\makeatother

\renewcommand{\leq}{\leqslant}	\renewcommand{\geq}{\geqslant}
\renewcommand{\bar}[1]{\overline{#1}}
\renewcommand\over[2]{{\,\buildrel #1\over#2\,}}

\newcommand{\inv}{^{-1}}

\newcommand {\limt}[2]{\xrightarrow[#1 \to #2]{}}

\newcommand{\abs}[1]{\left\vert #1\right\vert}        % valeur absolue
\newcommand{\nr}[1]{\left\Vert #1\right\Vert}         % norme
\newcommand{\innp}[2]{\left< #1 , #2 \right>}         % produit scalaire (inner product)  

\newcommand{\Dom}{\Dc}			% Domaine d'un opérateur
\newcommand{\pppg}[1] {\left< #1 \right>} 	% <x> = \sqrt{1+x^2}
\newcommand{\bigo}[2]{\mathop{O}\limits_{#1 \to #2}}

\newcommand{\singl}[1]{\left\{ #1 \right\}}		% Singleton --> en fait, n'importe quel ensemble
\newcommand{\Ii}[2] {\{#1,\dots,#2\}}         %{\llbracket #1,#2 \rrbracket}	% intervalle d'entiers.
\newcommand{\R}{\mathbb{R}}		\newcommand{\C}{\mathbb{C}}
\newcommand{\N}{\mathbb{N}}

\newcommand{\st}{\,:\,}					% ``tel que'' dans la définition d'un ensemble 

\newcommand{\seq}[2]{\left({#1}_{#2}\right)_{#2 \in\N}} % suite
\renewcommand{\Re}{\mathop{\rm{Re}}\nolimits}        % partie réelle
\renewcommand{\Im}{\mathop{\rm{Im}}\nolimits}        % partie imaginaire, Image
	
\DeclareMathOperator{\Id}{Id}                        % identité
 
\renewcommand{\a}{\alpha}\renewcommand{\b}{\beta}\newcommand{\g}{\gamma}\renewcommand{\d}{\delta}\newcommand{\D}{\Delta}\newcommand{\e}{\varepsilon}\newcommand{\z}{\zeta} \renewcommand{\th}{\theta}\newcommand{\Th}{\Theta}\renewcommand{\l}{\lambda}\newcommand{\n}{\nu}\newcommand{\s}{\sigma}\renewcommand{\t}{\tau}\newcommand{\f}{\varphi}\newcommand{\vf}{\phi}\newcommand{\h}{\chi}\newcommand{\p}{\psi}\renewcommand{\o}{\omega}\renewcommand{\O}{\Omega}

\newcommand{\Dc}{{\mathcal D}}\newcommand{\Hc}{{\mathcal H}}\newcommand{\Lc}{{\mathcal L}}\newcommand{\Pc}{{\mathcal P}}\newcommand{\Tc}{{\mathcal T}}

\newcommand{\loc}{{\rm{loc}}}

\newcounter{stepproof}
\newcommand{\stepp}{\stepcounter{stepproof} \noindent {\bf $\bullet$}\quad }

\newcommand{\Ha}{H_a}

\newcommand{\detail}[1]
{
}

\title{Exponential decay for the Schr\"odinger equation on a dissipative waveguide}

\author{Julien Royer}
\address{Institut de Math\'ematiques de Toulouse \\ 118, route de Narbonne \\ 31062 Toulouse C\'edex 09 \\ France}
\email{julien.royer@math.univ-toulouse.fr}

\begin{document}

\begin{abstract}

We prove exponential decay for the solution of the Schr\"odinger equation on a dissipative waveguide. The absorption is effective everywhere on the boundary but the geometric control condition is not satisfied. The proof relies on separation of variables and the Riesz basis property for the eigenfunctions of the transverse operator. The case where the absorption index takes negative values is also discussed.
\end{abstract}

\maketitle

\section{Introduction}

\newcommand{\oo}{l}
\newcommand{\dimm}{d}
\newcommand{\al} {a_\oo}
\newcommand{\ao} {a_0}
\newcommand{\Haoal}{H_{\al,\ao}}

Let $\oo > 0$ and $\dimm \geq 2$. Let $\O$ denote the straight waveguide $\R^{\dimm-1} \times ]0, \oo[ \subset \R^\dimm$. We consider on $\O$ the Schr\"odinger equation with dissipative boundary condition 
\begin{equation} \label{schrodinger}
\begin{cases}
i \partial_t u  + \D u  = 0 & \text{on } ]0,+\infty[ \times  \O,\\
\partial_\n u = i a  u  & \text{on } ]0,+\infty[ \times  \partial \O,\\
u(0,\cdot ) = u_0  & \text{on } \O .
\end{cases}
\end{equation}
Here $u_0 \in L^2(\O)$ and $\partial_\n$ denotes the outward normal derivative. The absorption index $a$ belongs to $W^{1,\infty}(\partial \O)$. In the main result of the paper, $a$ takes positive values, but we also discuss the case where $a$ takes (small) negative values.\\

We will first prove well-posedness for this problem. Then it is standard computation to check that when $a \geq 0$ the norm of $u(t)$ in $L^2(\O)$ is non-increasing, and that the decay is due to the boundary condition:
\begin{align*}
\frac d {dt} \nr{ u(t)}_{L^2(\O)}^2 = - 2  \int_{\partial \O} a \abs {u(t)}^2  \leq 0.
\end{align*}
Whether this norm goes to 0 for large times, and then the rate of decay, are questions which have been extensively studied in different contexts. For the Schr\"odinger equation as in the present paper, or for the (damped) wave equation which is a closely related problem.\\

Many papers deal with the wave equation on compact manifolds, with dissipation in the interior of the domain or at the boundary.

We know from \cite{haraux85,lebeau96} that a weak assumption on the absorption index $a$ (for instance the dissipation is effective on any open subset of the domain) is enough to ensure that the energy goes to 0 for any initial datum. 

Uniform exponential decay has been obtained in \cite{raucht74,bardoslr92} under the now usual geometric control condition. Roughly speaking, the assumption is that any (generalized) bicharacteristic (or classical trajectory, or ray of geometric optics) meets the damping region (in the interior of the domain or at the boundary). For the free wave equation on a subset of $\R^\dimm$, the spatial projections of these bicharacteristics are straight lines, reflected at the boundary according to classical laws of geometric optics. This condition is essentially necessary and sufficient (we do not discuss here the subtilities due to the trajectories which meet the boundary tangentially).

Then the question was to understand what happens when this damping condition fails to hold. In \cite{lebeau96, lebeaur97} it is proved that we have at least a logarithmic decay of the energy if the initial datum belongs to the domain of the infinitesimal generator of the problem. This can be optimal, in particular when non controlled trajectories are stable. Intermediate rates of decay have been obtained for several examples where the flow is unstable near these trajectories (see for instance \cite{liur05,Burq-Hi-07,christianson07er,schenck11,anantharamanl}).\\

The same questions have been investigated for the Schr\"odinger or wave equation on the Euclidean space. Even in the self-adjoint case, where the norm for the solution of the Schr\"odinger equation or the energy for the solution of the wave equation are conserved, it is interesting to study the local energy decay, which measures the fact that the energy escapes at infinity. For the free case and with compactly supported initial conditions, explicit computations show that the energy on a compact vanishes after finite time for the wave in odd dimension $d\geq 3$, while the decay is of size $t^{-d}$ in even dimension. For the Schr\"odinger equation, the norm of the solution decays like $t^{-\frac \dimm 2}$.

Many authors have proved similar estimates for perturbed problems. For instance the wave or Schr\"odinger equations can be stated outside a compact obstacle of $\R^\dimm$ (we can also consider a perturbation of the Laplace operator). For the wave equation on an exterior domain, we have uniform (exponential) decay for the local energy if and only if the non-trapping assumption holds. This means that there is no trapped bicharacteristic. See \cite{lax-phillips,ralston69,melrose79}. See \cite{burq98} for logarithmic decay without the non-trapping assumption, and \cite{nonnenmacherz09} for an intermediate situation.

For the dissipative equations, this non-trapping assumption can be replaced by the same damping condition (geometric control condition) on bounded trajectories as in the compact case (where all the trajectories were bounded). This means that the energy of the wave (or at least the contribution of high frequencies) escapes at infinity or is dissipated by the medium. This has been used in \cite{alouik07} for a dissipation in the interior of the domain and in \cite{aloui02,alouik10} for dissipation at the boundary, for the free 
equations on an exterior domain in both cases. See also \cite{art-mourre} for the corresponding resolvent estimates, \cite{boucletr14} for the damped wave equation with a Laplace-Beltrami operator corresponding to a metric which is a long-range perturbation of the flat one, and \cite{alouikv13} where the damping condition is not satisfied but the dissipation is stronger.\\

In this paper we consider a domain which is neither bounded nor the complement of a bounded obstacle. In particular, compared to the situations mentioned above, the boundary of the wave\-guide is not compact.

More precisely, we are going to use the fact that $\O$ is a Euclidean space in some directions and compact with respect to the last coordinate, so that properties of both compact and Euclidean domains will appear in our analysis.\\

The main result of this paper is the following:

\begin{theorem} \label{th-energy-decay}
Assume that there exist two constants $a_0,a_1$ such that on $\partial \O$ we have 
\begin{equation} \label{hyp-a0-a-a1}
0 < a_0 \leq a \leq a_1. 
\end{equation}
Then there exist $\g > 0$ and $C \geq 0$ such that for all $u_0 \in L^2(\O)$ the solution $u$ of the problem \eqref{schrodinger} satisfies
\[
\forall t \geq 0, \quad \nr{u(t)}_{L^2(\O)} \leq C e^{-\g t} \nr{u_0}_{L^2(\O)}.
\]
The same result holds if $a$ vanishes on one side of the boundary and satisfies \eqref{hyp-a0-a-a1} on the other side.
\end{theorem}

In this theorem the absorption is effective at least on one side of the boundary, so the damping condition is clearly satisfied on bounded trajectories (which have to meet both sides of $\partial \O$). However, it is important to note that we prove exponential decay for the total energy and not only for the local energy. This means that all the energy is dissipated by the medium, including the energy going at infinity. This suggests that all the classical trajectories, and not only the bounded ones, should be controlled by the dissipation. This is not the case here.

Thus this theorem provides a new example for the already mentioned general idea that given a result for which the non-trapping condition (or the damping condition for dissipative problems) is necessary, we get a close result if there are only a few trajectories which contradict the assumption. Here the initial conditions in $T^* \O$ whose corresponding trajectories avoid the boundary form a submanifold of codimension 1 (the frequency vector has no tranversal component). Moreover the flow is linearly unstable. Compared to the related results on compact or Euclidean domains, our analysis will be based on simpler arguments. The key argument is that our problem will inherit the decay property of the transverse problem (that is the Schr\"odinger equation on $]0,\oo[$ with the same dissipative boundary condition), for which the geometric control assumption holds.\\

In order to prove time-dependent estimates for an evolution equation as in Theorem \ref{th-energy-decay}, we often use spectral properties for the generator of the problem (spectral gap for the eigenvalues on a compact manifold, absence of resonances close to a positive energy on a perturbation of the Euclidean space, etc.). 
Here, the problem is governed by the operator
\begin{equation} \label{def-Ha}
\Ha \f = -\D \f
\end{equation}
defined on the domain
\begin{equation} \label{dom-Ha}
\Dom(\Ha) = \singl {\f \in H^2(\O) \st \partial_\n \f = i a \f  \text { on } \partial \O} \subset L^2(\O).
\end{equation}
In particular, when $a \geq 0$ the solution $u$ of the problem \eqref{schrodinger} is given by the semi-group $e^{-it\Ha}$ generated by this operator. When $u_0 \in \Dom(\Ha)$ this solution belongs to $C^0(\R_+ , \Dom(\Ha)) \cap C^1(\R_+,L^2(\O))$.\\

The main part of this paper is devoted to the proof of resolvent estimates for the operator $\Ha$. More precisely, in order to prove Theorem \ref{th-energy-decay} we will use the following result:

\begin{theorem} \label{th-gap-spectral}
Let $a$ be as in Theorem \ref{th-energy-decay}. Then there exist $\tilde \g > 0$ and $C \geq 0$ such that any $z \in \C$ with $\Im (z) \geq - \tilde \g$ is in the resolvent set of $\Ha$ and moreover
\[
\nr{ (\Ha - z)\inv  }_{\Lc(L^2(\O))} \leq C.
\]
\end{theorem}

Here $\Lc(L^2(\O))$ stands for the set of bounded operators on $L^2(\O)$.
When $a$ is constant, our analysis will provide a good description of the spectrum of $\Ha$, and in particular there is a spectral gap as stated in Theorem \ref{th-gap-spectral}. However, for a non-selfadjoint operator the norm of the resolvent can be large far from the spectrum, so the uniform bound of the resolvent on a strip around the real axis is not a consequence of the spectral gap and has to be proved directly.\\

In the papers we have mentioned above concerning the damped Schr\"odinger equation \cite{aloui08,aloui08b,alouikv13,alouikr}, it is proved that this equation satisfies the Kato smoothing effect, that is an estimate of the form 
\[
\int_0^\infty \nr{w(x) (1-\D)^\frac 14  u(t)}_{L^2}^2 \, dt \lesssim \nr {u_0}_{L^2}^2,
\]
for some suitable weight function $w$. In these papers, the absorption is in the interior of the domain and, more important, it is of the form $a(x) (1-\D)^{\frac 12} a(x)$. This is stronger than multiplication by $a(x)^2$ for high-frequencies. It is not the case here, which is why we have no smoothing effect for our problem. However, we can at least recover the same result as in the non-selfadjoint case $a=0$ (see for instance \cite{danconar12}), namely the smoothing property in the unbounded directions:

\begin{theorem} \label{th-smoothing}
Let $a$ be as in Theorem \ref{th-energy-decay}. Let $\d > \frac 12$. Then there exists $C \geq 0$ such that for all $u_0 \in L^2(\O)$ the solution $u$ of \eqref{schrodinger} satisfies 
\[
\int_0^\infty \nr{\pppg x ^{-\d} \big( 1- \D_x\big)^{\frac 14} u(t)}_{L^2(\O)}^2 \, dt \leq C \nr {u_0}_{L^2(\O)}^2,
\]
where $\D_x = \sum_{n=1}^{\dimm-1} \partial_{x_n}^2$ is the partial Laplacian acting on the unbounded directions.
\end{theorem}

Here and everywhere below we write $\pppg x$ for $\big( 1 + \abs x^2 \big)^{\frac 12}$. The proof of Theorem \ref{th-smoothing} relies on the theory of relatively smooth operators in the sense of Kato. In order to apply it in our non-selfadjoint setting we will use a self-adjoint dilation of our dissipative operator.\\

In the results above, we can relax the assumption that the absorption index is positive everywhere. In Theorem \ref{th-energy-decay} we consider the case where $a$ is positive on one side of the boundary, say $\R^{\dimm-1} \times \singl \oo$, and vanishes on the other side $\R^{\dimm-1} \times \singl 0$. What happens if the absorption index is a negative constant on one side is not so clear.\\

% We now assume that the absorption index $a$ is equal to $\ao$ on $\R^{\dimm-1} \times \singl 0$ and equal to $\al$ on $\R^{\dimm-1} \times \singl l$, where $\al,\ao \in \R$ are constants which satisfy
Let $\al,\ao \in \R$ be such that
\begin{equation} \label{sum-ab}
\al + \ao > 0
\end{equation}
and consider the problem
% The Schr\"odinger equation \eqref{schrodinger} now reads
%
\begin{equation} \label{schrodinger-nondiss}
\begin{cases}
i \partial_t u  + \D u  = 0 & \text{on } ]0,+\infty[ \times  \O,\\
\partial_\n u = i \al  u  & \text{on } ]0,+\infty[ \times \R^{\dimm-1} \times  \singl \oo,\\
\partial_\n u = i \ao  u  & \text{on } ]0,+\infty[ \times  \R^{\dimm-1} \times  \singl 0,\\
u(0,\cdot ) = u_0  & \text{on }  \O.
\end{cases}
\end{equation}
We denote by $\Haoal$ the corresponding operator: $\Haoal = -\D$ with domain
\[
\Dom(\Haoal) = \singl { u \in H^2(\O) \st \partial_\n u = i \al  u  \text{ on } \R^{\dimm-1} \times  \singl \oo, \partial_\n u = i \ao  u   \text{ on }   \R^{\dimm-1} \times  \singl 0}.
\]

When $\al$ and $\ao$ have different signs but satisfy \eqref{sum-ab}, we will say that the boundary condition is weakly dissipative. Estimates like those of Theorems \ref{th-energy-decay} and \ref{th-gap-spectral} are not likely to hold without this assumption. For instance, when $\ao = - \al$ we get a $\Pc \Tc$-symmetric waveguide. Such a boundary condition has been studied in \cite{krejcirikbz06,borisovk08}. See also \cite{krejcirikt08}. In particular it is known that in this case the spectrum is real, so we cannot expect any generalization of our results. The case $\al+\ao < 0$ is dual to the case $\al + \ao > 0$. We will see that in this case the norm of the solution grows exponentially (see Remark \ref{rem-anti-diss}).\\

In the weakly dissipative case we have the following theorem:

\begin{theorem} \label{th-energy-decay-nondiss}
Let $\al, \ao \in \R$ satisfy Assumption \eqref{sum-ab}. If $u_0 \in \Dom(\Haoal)$ then the problem \eqref{schrodinger-nondiss} has a unique solution $u \in C^1(]0,+\infty[,L^2(\O)) \cap C^0 ([0,+\infty[,\Dom(\Haoal))$. Moreover, if $\abs {\al}$ and $\abs {\ao}$ are small enough there exist $\g > 0$ and $C \geq 0$ such that for all $u_0$ we have
\[
\forall t \geq 0, \quad \nr{u(t)}_{L^2(\O)} \leq C e^{-\g t} \nr{u_0}_{L^2(\O)}.
\]
\end{theorem}

Let us discuss the assumption that the absorption has to be small. In \cite{art-nondiss} it was proved in another context that for high frequencies the properties of a dissipative problem remain valid when the absorption index is positive on average along the corresponding classical flow. This is the case here under Assumption \eqref{sum-ab}. The restriction comes from low frequencies. When $a \neq 0$, the boundary condition can be rewritten 
\[
u = \frac 1 {ia} \partial_\n u \quad \text{on } \partial \O.
\]
When $a$ is large compared to the frequency, this is close to a Dirichlet condition. That for large dissipation we recover a self-adjoint problem is usually called the overdamping phenomenon. This suggests that the problem is now governed by the quantity $1/a$. And when $\al \ao < 0$, Assumption \eqref{sum-ab} can be rewritten as 
\[
 \frac 1 {\al}  + \frac 1 {\ao} < 0.
\]
It turns out that for $\al,\ao$ large enough with $\al \ao < 0$ and \eqref{sum-ab}, the contribution of low frequencies is indeed exponentially increasing (see Proposition \ref{prop-imvp-up}).\\

In this paper we prove all these results on the model case of a straight waveguide with a one-dimensional section. The purpose is on one hand to observe all the non-trivial phenomenons mentioned above on a quite simple example. On the other hand our analysis is the first step toward the understanding of similar properties for the wave equation on a more general domain.

On this model case, and with the additional assumption that $a$ is constant, we can rewrite $\Ha$ as the sum of the usual Laplacian on the Euclidean space $\R^{\dimm-1}$ and a Laplace operator on the compact section. Since this section is of dimension 1, we can give quite explicitely many spectral properties for this operator. In particular, we will see that its eigenfunctions form a Riesz basis, which will give a good description of the spectrum of $\Ha$, first when $a$ is constant and then in the general case. The time dependent estimate will follow.\\

In Section \ref{sec-operator} we prove that if $a \geq 0$ then $\Ha$ is maximal dissipative, which gives in particular well-posedness for the problem \eqref{schrodinger}. In Section \ref{sec-transverse} we study the transverse operator, that is the Laplace operator on the section $]0,\oo[$. Spectral properties of $\Ha$ are obtained for a constant absorption index in Section \ref{sec-separation}, and then Theorems \ref{th-gap-spectral} and \ref{th-smoothing} are proved in Section \ref{sec-a-variable}. Once the spectral properties of $\Ha$ are well-understood, the proof of Theorem \ref{th-energy-decay} is given in Section \ref{sec-time-decay}. Finally, the problem \eqref{schrodinger-nondiss} where $a$ can take negative values is discussed in Section \ref{sec-non-diss}.\\

All along the paper, a general point in $\O$ will be written $(x,y) \in \O \simeq \R^{\dimm-1} \times ]0,\oo[$, with $x \in \R^{\dimm-1}$ and $y \in ]0,\oo[$. As in Theorem \ref{th-smoothing}, we denote by $\D_x$ the usual laplacian on $\R^{\dimm-1}$. For Hilbert spaces $\Hc_1$ and $\Hc_2$, $\Lc(\Hc_1,\Hc_2)$ is the set of bounded operators from $\Hc_1$ to $\Hc_2$. For $\g > 0$ we finally set
\[
\C_\g = \singl{z \in \C \st \Im z > -\g}, \quad \text{and then} \quad \C_+ := \C_0.
\]

\bigskip 

\noindent  
{\bf Acknowledgements: } I am grateful to Petr Siegl for stimulating discussions which motivated this paper and helped me through its realization. This work is partially supported by the French-Czech BARRANDE Project 26473UL and by the French National Research Project NOSEVOL (ANR 2011 BS01019 01).

\section{Operator associated to the dissipative waveguide} \label{sec-operator}

\newcommand{\bHa}{\hat H_a}

In this section we consider a more general waveguide $\O$ of the form $\R^{p}\times \o \subset \R^\dimm$ where $p \in \Ii 1 {\dimm-1}$ and $\o$ is a smooth open bounded subset of $\R^{\dimm-p}$. In particular, $\O$ is open in $\R^\dimm$.\\

Let $a \in W^{1,\infty}(\partial \O)$. Until Proposition \ref{prop-max-diss}, we make no assumption on the sign of $a$. We consider on $L^2(\O)$ the operator $\Ha$ defined by \eqref{def-Ha} with domain \eqref{dom-Ha}.
For all $\f \in \Dom(\Ha)$ we have
\begin{equation} \label{eq-Ha-forme}
\innp{\Ha \f} \f_{L^2(\O)} = - \innp{\D \f}{\f}_{L^2(\O)} = \innp{\nabla \f}{\nabla \f}_{L^2(\O)} -i \innp{a \f}{\f}_{L^2(\partial \O)}.
\end{equation}
On the other hand, we consider the quadratic form defined for $\f \in \Dom(q_a) = H^1(\O)$ by
\[
q_a(\f) = \int_\O \abs {\nabla \f}^2 - i \int_{\partial \O} a\abs \f^2 .
\]
We also denote by $q_a$ the corresponding sesquilinear form on $\Dom(q_a)^2$. That this quadratic form is sectorial and closed follows from the following lemma and traces theorems:

\begin{lemma} \label{lem-sect-form}
Let $q_R$ be a non-negative, densely defined, closed form on a Hilbert space $\Hc$. Let $q_I$ be a symmetric form relatively bounded with respect to $q_R$. Then the form $q_R - i q_I$ is sectorial and closed.
\end{lemma}

It is important to note that there is no smallness assumption on the relative bound of $q_I$ with respect to $q_R$. In particular, for $q_a$ we do not need any assumption on the size of $a$ in $L^\infty(\partial \O)$.

\begin{proof}
There exists $C >0$ such that for all $\f \in \Dom(q_R)$ we have 
\[
\abs{q_I(\f)} \leq C \left( q_R(\f) + \nr{\f}^2_\Hc\right).
\]
Let $\e_0 = \frac 1 {2C}$. If we already know that $(q_R -i\l q_I)$ is sectorial and closed for some $\l \geq 0$, then $(q_R -i(\l+\e) q_I)$ is sectorial and closed for all $\e \in [0,\e_0]$ according to Theorem VI.3.4 in \cite{kato}. Now since $q_R$ is sectorial and closed, we can prove by induction on $n \in \N$ that $(q_R -i\l q_I)$ is sectorial and closed for all $\l \in [0,n\e_0]$, and hence for all $\l \geq 0$. This is in particular the case when $\l=1$.
\end{proof}

We recall the definitions of accretive and dissipative operators (note that the conventions may be different for other authors):

\begin{definition}
We say that an operator $T$ on the Hilbert space $\Hc$ is accretive (respectively dissipative) if
\[
\forall \f \in \Dom(T), \quad \Re \innp{T \f}\f_\Hc \geq 0, \quad \big(\text{respectively} \quad \Im \innp{T \f}\f _\Hc \leq 0\big).
\]
Moreover $T$ is said to be maximal accretive (maximal dissipative) if it has no other accretive (dissipative) extension on $\Hc$ than itself. In particuliar $T$ is (maximal) dissipative if and only if $iT$ is (maximal) accretive.
\end{definition}

Let us recall that an accretive operator $T$ is maximal accretive if and only if $(T-z)$ has a bounded inverse on $\Hc$ for some (and hence any) $z \in \C$ with $\Re(z) < 0$. In this case we know from the Hille-Yosida Theorem that $-T$ generates a contractions semi-group $t \mapsto e^{-tT}$. Then for all $u_0 \in \Dom(T)$ the map $u : t \mapsto e^{-tT}u_0$ belongs to $C^1(\R_+,\Hc) \cap C^0(\R_+,\Dom(T))$ and solves the problem 
\[
\begin{cases}
u'(t) + T u(t) = 0, \quad \forall t > 0,\\
u(0) = u_0.
\end{cases}
\]

Let us come back to our context. According to Lemma \ref{lem-sect-form} and the Representation Theorem VI.2.1 in \cite{kato}, there exists a unique maximal accretive operator $\bHa$ on $L^2(\O)$ such that $\Dom(\bHa) \subset \Dom(q_a)$ and 
\[
\forall \f \in \Dom(\bHa), \forall \p \in \Dom(q_a) , \quad \innp{\bHa \f} {\p}_{L^2(\O)} = q_a(\f,\p).
\]
Moreover we have
\[
\Dom(\bHa) = \singl{u \in \Dom(q_a) \st \exists f \in L^2(\O) , \forall \vf \in \Dom(q_a), q_a(u,\vf) = \innp{f}{\vf}},
\]
and for $u \in \Dom(\bHa)$ the corresponding $f$ is unique and given by $\bHa u = f$.

\begin{proposition} \label{prop-sectorial}
We have $\bHa = \Ha$. In particular $\Ha$ is maximal accretive.
\end{proposition}

For a one-dimensional section we can essentially follow the proof of Lemma 3.2 in \cite{borisovk08}. This would be enough for our purpose but, for further use, we prove this result in the general setting.

\begin{proof}
It is easy to check that $\Dom(\Ha) \subset \Dom(\bHa)$ and $\Ha = \bHa$ on $\Dom(\Ha)$. Now let $u \in \Dom(\bHa)$.
By definition there exists $f \in L^2(\O)$ such that 
\[
\forall \vf \in H^1(\O), \quad \int_{\O} \nabla u \cdot\nabla \bar {\vf} -i \int_{\partial \O} a u \bar \vf = \int_{\O} f \bar \vf.
\]
Considering $\vf \in C_0^\infty(\O)$ we see that $-\D u = f$ in the sense of distributions and hence in $L^2(\O)$. This proves that $u \in H^2_\loc(\O)$. It remains to prove that $u \in H^2(\O)$ and that the boundary condition $\partial_\n u = i a u$ holds on $\partial \O$.
Let $j \in \Ii 1 {p}$ and let $e_j$ be the $j$-th vector in the canonical basis of $\R^{p}$. Let $\d \in \R^*$ and $u_\d : (x,y) \mapsto  \frac 1 \d ( u(x + \d e_j , y ) - u(x,y)) \in H^1(\O)$. See for instance \cite[\S 5.8.2]{evans} for the properties of the difference quotients. For all $\vf \in H^1(\O)$ we have
\begin{align*}
q_a(u_\d , \vf) = -\int_\O f(x,y) \bar {\vf_{-\d}} (x,y) \, dx \,d\s(y) - i  \int_{\partial \O} u (x,y) a_{-\d}(x,y) \bar \vf (x-\d e_j,y) \, dx \, d\s(y),
\end{align*}
where $\s$ is the Lebesgue measure on $\o$. Since $a \in W^{1,\infty}(\partial \O)$ there exists $C \geq 0$ such that for all $\vf \in H^1(\O)$ and $\d > 0$ we have
\[
\abs{q_a (u_\d, \vf)} \leq C \nr {\vf}_{H^1(\O)}.
\]
Applied with $\vf = u_\d$ this gives 
\[
\nr{u_\d}_{\dot H^1(\O)}^2 = \Re q_a (u_\d,u_\d) \leq C \nr {u_\d}_{H^1(\O)}.
\]
Since we already know that $u_\d \in L^2(\O)$ uniformly in $\d > 0$, we have
\[
\nr{u_\d}_{H^1(\O)}^2  \lesssim 1 + \nr {u_\d}_{H^1(\O)},
\]
which implies that $u_\d$ is uniformly in $H^1(\O)$. This means that $\partial _{x_j}  u \in H^1(\O)$. Since this holds for any $j \in \Ii 1 {p}$, this proves that all the derivatives of order 2 with at least one derivative in the first $p$ directions belong to $L^2(\O)$. Then we get
\[
-\D_y u = f + \D_x u \in L^2(\O).
\]
According to the Green Formula we have for all $\vf \in H^1(\O)$
\begin{align*}
\int_{\O} - \D_y u \,  \bar \vf \, dx \, dy = \int_{\O} \nabla_y u \cdot \nabla_y \bar \vf \, dx \, dy -  \innp{\partial_\n u}{\vf}_{H^{-1/2}(\partial \O),H^{1/2}(\partial \O)}
\end{align*}
(see for instance \cite[eq. (1.5.3.10)]{grisvard}). By density of the trace map, we obtain that 
\begin{equation} \label{cond-bord}
\partial _\n u = ia u \quad \text{on } \partial \O.
\end{equation}
In particular $\partial_\n u \in H^{1/2}(\partial \O)$. Then there exists $v \in H^2(\O)$ such that $\partial_\n v = \partial _\n u$ (see \cite[Th. 1.5.1.1]{grisvard} for a fonction on $\R^\dimm_+$ ; for a function on $\O$ we follow the same idea as for a fonction on a bounded domain, except that we only use a partition of unity for $\partial \o$, which allows to cover $\partial \O$ by a finite number of strips, each of which is diffeomorphic to a strip on $\R^{p}$). Let $w = u - v$. We have $w \in H^1(\O)$, $\D_y w \in L^2(\O)$ and 
\[
\begin{cases}
-\D_y w + w = f + \D_x u - \D_y v + w & \text{on } \O, \\
 \partial_\n w = 0& \text{on } \partial \O.
\end{cases}
\]
Then for almost all $x \in \R^{p}$ we have 
\[
\begin{cases}
-\D_y w(x) + w(x) = f(x) + \D_x u(x) - \D_y v(x) + w(x) & \text{on } \o, \\
 \partial_\n w(x) = 0& \text{on } \partial \o.
\end{cases}
\]
By elliptic regularity for the Neumann problem (see for instance Theorem 9.26 in \cite{brezis}) we obtain that $w(x) \in H^2(\o)$ with 
\[
\nr{w(x)}_{H^2(\o)} \lesssim \nr{f(x) + \D_x u(x) - \D_y v(x) + w(x)}_{L^2(\o)}.
\]
After integration over $x \in \R^{p}$, this gives 
\[
\nr{u}_{H^2_y(\O)} \lesssim \nr{f+ \D_x u - \D_y v + w}_{L^2(\O)} + \nr{v}_{H^2(\O)}.
\]
Since we already know that second derivatives of $u$ involving a derivation in $x$ are in $L^2(\O)$, this proves that $u \in H^2(\O)$ and concludes the proof.
\end{proof}

\begin{remark} \label{rem-H-a}
We have $\Ha^* = H_{-a}$.
\end{remark}

Now assume that $a$ takes non-negative values. According to \eqref{eq-Ha-forme}, $H_a$ is a dissipative operator. Since it is maximal accretive, it is easy to see that it is in fact maximal dissipative:

\begin{proposition} \label{prop-max-diss}
The maximal accretive operator $\Ha$ is also maximal dissipative.
\end{proposition}

\begin{proof}
We already know that $\Ha$ is dissipative. Since it is maximal accretive, any $z \in \C$ with $\Re z < 0$ is in its resolvent set. Then it is easy to find $z$ in the resolvent set of $\Ha$ with $\Im z > 0$.
\end{proof}

\begin{proposition} \label{prop-realvp}
If $a > 0$ in an open subset of $\partial \O$ then $\Ha$ has no real eigenvalue.
\end{proposition}

\begin{proof}
Let $u \in \Dom(\Ha)$, $\l \in \R$, and assume that $\Ha u = \l u$. Taking the imaginary part of the equality $q_a (u,u) = \l \nr u ^2$ gives 
\[
\int_{\partial \O} a \abs{u}^2 = 0.
\]
This implies that $u = 0$ where $ a \neq 0$ and $\partial _\n u = ia u = 0$ everywhere on $\partial \O$. By unique continuation, this implies that $u = 0$ on $\O$.
\end{proof}

\section{The Transverse Operator} \label{sec-transverse}

\newcommand{\Ta} {T_a}
\newcommand{\Tam} {T_{a_m}}
\newcommand{\HoI} {T_0}
\newcommand{\srev}{\l}

Let us come back to the case of a one-dimensional cross-section $\o = ]0,\oo[$. Under the additional assumption that the absorption index $a$ is constant on $\partial \O$ the operator $\Ha$ can be written as 
\begin{equation} \label{separation-variable-Ha}
\Ha = -\D_x \otimes \Id_{L^2(0,\oo)} + \Id _{L^2(\R^{\dimm-1})} \otimes  \Ta,
\end{equation}
where $-\D_x$ is as before the usual flat Laplacian on $\R^{\dimm-1}$ and $\Ta$ is the transverse Laplacian on $]0,\oo[$. More precisely, we consider on $L^2(0,\oo)$ the operator $\Ta = - \frac {d^2} {dy^2}$ with domain
\[
\Dom(\Ta) = \singl{u \in H^2(0,\oo) \st u'(0) = -ia u(0), u'(\oo) = ia u(\oo)}.
\]
This is the maximal accretive and dissipative operator corresponding to the form 
\[
q : u \in H^1(0,\oo) \mapsto \int_{0}^{\oo} \abs{u'(x)}^2 \, dx - i a \abs{u(\oo)}^2 - i a \abs{u(0)}^2.
\]

In this section we give the spectral properties of $\Ta$ which we need to study the full operator $\Ha$. This operator has compact resolvent, and hence its spectrum is given by a sequence of isolated eigenvalues. When $a = 0$, which corresponds to the Neumann problem, we know that the eigenvalues of $T_0$ are the real numbers $n^2 \n^2$ for $n \in \N$, where we have set
\[
\n = \frac \pi { \oo}.
\]
These eigenvalues are algebraically simple.

\begin{proposition} \label{prop-srev}
There exists a sequence $\seq \srev n$ of continuous functions on $\R$ such that $\srev_n(0) = n \n$ and for all $a \in \R$ the set of eigenvalues of $\Ta$ is $\singl{\srev_n(a)^2, n\in\N}$. Moreover:
\begin{enumerate}[(i)]
\item For $(n,a) \in (\N \times \R) \setminus \{(0,0)\}$ the eigenvalue $\srev_n(a)^2$ is algebraically simple and a corresponding eigenvector is given by
\begin{equation} \label{def-phi}
\f_{n}(a) : x \mapsto   A_n(a)  \left( e^{i\srev_n(a) x} + \frac {\srev_n(a) + a}{\srev_n(a) - a} e^{-i \srev_n(a) x} \right) ,
\end{equation}
where we can choose $A_n(a) \in \R_+^*$ in such a way that $\nr{\f_n(a)}_{L^2(0,\oo)} = 1$ (when $a = 0$ then 0 is a simple eigenvalue and corresponding eigenvectors are non-zero constant functions).
\item For $n \in \N$ and $a \in \R$ we have $\srev_n(-a)= \bar {\srev_n(a)}$.

\item Let $n\in\N$. For all $a \in \R^*$ we have $\Re (\l_n(a)) \in ] n\n,(n+1)\n[$ (when $n=0$, we have chosen the square root of $\srev_0^2(a)$ which has a positive real part).% and the map $a \mapsto \Re(\l_n(a))$ is increasing on $\R_+$.
\item For all $n \in \N$ there exists $C_n > 0$ such that for $a > 0$ we have $-C_n < \Im (\srev_n(a)) < 0$. 
\item 
Let $a > 0$ be fixed. We have 
\[
\srev_n(a) = n \n - \frac {2ia}{n\n \oo} + \bigo n {+\infty} \big(n^{-2} \big)
\]
and hence
\[
\srev_n(a)^2 = (n\n)^2 - \frac {4ia}{\oo} + \bigo n {+\infty} \big(n^{-1} \big).
\]
\end{enumerate}
\end{proposition}

\begin{proof}

\stepp It is straightforward computations to check that 0 is an eigenvalue of $\Ta$ if and only if $a = 0$ and, if $\srev \in \C^*$, $\l^2$ is an eigenvalue if and only if 
\begin{equation} \label{eq-z}
(a-\srev )^2 e^{2i\srev \oo} = ( a+\srev )^2.
\end{equation}
If $\l^2$ is an eigenvalue then the corresponding eigenfunction is of the form $\f : x \mapsto  A e^{i\srev x} + B e^{-i\srev x}$ with 
\begin{equation} \label{lien-AB}
A = \frac {\srev -a}{\srev +a} B = \frac {\srev + a}{\srev -a} e^{-2i\srev \oo}B.
\end{equation}
Moreover, all these eigenvalues have geometric multiplicity 1. Indeed, given $n \in \N$, the space of eigenvectors corresponding to the eigenvalue $\srev_n(a)^2$ is strictly included in the space of $H^2$ functions which are solutions of $-u'' - \srev_n(a)^2 u = 0$, and this space is of dimension 2. The fact that the eigenvalues of $H_{-a}$ are conjugated to the eigenvalues of $H_{a}$ is a consequence of Remark \ref{rem-H-a}.

\stepp Let $a > 0$ and $\srev \in \C^*$ be such that $\srev^2$ is an eigenvalue of $\Ta$. Assume that $\Re\srev \in \n \N$. Then
\[
\left( \frac {a+\l}{a-\l} \right)^2 = e^{2i\l\oo} \in \R_+
\]
(note that $\srev$ cannot be equal to $a$ in \eqref{eq-z}) and hence 
\[
r := \frac {a+\l}{a-\l} \in \R.
\]
If $r = -1$ then $a = 0$. Otherwise $\l =  \frac{a(1-r)}{1+r} \in \R$. In both cases we obtain a contradiction (see Proposition \ref{prop-realvp}), and hence $\Re\srev \notin \n \N$. This proves that for $a > 0$ the operator $\Ta$ has no eigenvalue with real part in $\n \N$.

\stepp Now let $R > 0$. We prove that if $C_R \geq 0$ is chosen large enough and if $a \in\R$ and $\srev \in \C^*$ are such that $\srev^2$ is an eigenvalue of $\Ta$, then 
\begin{equation} \label{eq-ImRe-srev}
\abs{\Re \srev} \leq R \quad \implies \quad \abs {\Im \srev} \leq C_R.
\end{equation}
Assume by contradiction that this is not the case. Then for all $m \in \N$ we can find $x_m \in [-R,R]$ and $y_m \in \R$ with $\abs {y_m} \geq m$ such that $(x_m+ i y_m)^2$ is an eigenvalue of $\Tam$ for some $a_m \in\R$. We have 
\[
e^{-2y_m \oo} = \abs{\frac {a_m + x_m + iy_m}{a_m - x_m - iy_m}}^2 = \frac {(a_m + x_m)^2 + y_m^2}{(a_m - x_m)^2 + y_m^2} \limt m \infty 1,
\]
which gives a contradiction.

\stepp The family of operators $\Ta$ for $a \in \R$ is an analytic family of operators of type B in the sense of Kato \cite[\S VII.4]{kato}. We already know that the spectrum of $\HoI$ is given by $\singl{(n\n)^2, n \in \N}$. Then for all $n \in \N$ there exists an analytic function $\l_n^2$ such that, at least for small $a$, $\l_n^2(a)$ is in the spectrum of $\Ta$ (and then we define $\srev_n$ as the square root of $\srev_n^2$ with positive real part). See Theorem VII.1.7 in \cite{kato}.

\stepp
Let $n \in \N^*$. We write $\srev_n(a) = n\n + \b a + \g a^2 + O_{a \to 0} (a^3)$. We have
\[
e^{2i \srev_n(a) \oo} 
%= e^{4i\oo \left(\b a + \g a^2 + O (a^3)\right)} 
= 1 + 2i\oo \b a + 2i\oo \g a^2 -  2\oo^2 \b^2 a^2 + O(a^3),
\]
and on the other hand:
\begin{align*}
\left( \frac {\srev_n(a) + a} {\srev_n(a) - a} \right)^2
= 1 + \frac {4a}{n\n} -  \frac {4(\b-2) a^2} {n^2\n^2} + O\big(a^3\big).
\end{align*}
\noindent 
Since $\srev_n(a)$ solves \eqref{eq-z} for any $a > 0$ we obtain
\[
\b = \frac 2 {i\oo n\n} = -\frac {2i}{\pi n}
\]
and
\[
\Re (\g)
= \frac {4 \oo}  {n^3\pi^3} .
\]
Since $\Re (\b) = 0$ and $\Re (\g) > 0$ we have $\Re \big(\srev_n(a)\big) \in \left] n\n , (n+1)\n \right[$ for $a > 0$ small enough. The functions $a \mapsto \Re \big(\srev_n(a)\big)$ are continuous and cannot reach $\n \N$ unless $a=0$, so this remains true for any $a > 0$ such that $\srev_n(a)$ is defined. Similarly $\Re \big(\srev_0(a)\big) \in ]0, \n[$ for all $a > 0$. In particular the curves $a \mapsto \srev_n(a)$ for $n \in \N$ never meet. Moreover we know from \eqref{eq-ImRe-srev} that $\srev_n(a)$ remains in a bounded set of $\C$, so the curves $a \mapsto \srev_n(a)$ are defined for all $a \in \R$ and for all $a \in\R$ the eigenvalues of $\Ta$ are exactly given by $\srev_n(a)^2$ for $n \in \N$. 

\stepp It remains to prove that the asymptotic expansion of $\srev_n(a)$ for $n$ fixed and $a$ small is also valid for $a$ fixed and $n$ large. Let $a>0$ be fixed. Derivating \eqref{eq-z} and using the fact that $\abs{\srev_n(s)} \geq n \n$ for all $s \in \R$ we obtain that 
\[
\sup_{s \in [0,a]} \abs{\srev'_n(s)} = O \big(n\inv\big).
\]
This means that $\srev_n(a) = n\n + O(n\inv)$. Then we obtain the asymptotic expansion of $\srev_n(a)$ for large $n$ as before, using again \eqref{eq-z}. This gives the last statement of the proposition and concludes the proof. 
\end{proof}

Now that we have proved what we need concerning the spectrum of the operator $\Ta$, we study the corresponding sequence of eigenfunctions. In the self-adjoint case $a = 0$, we know that the eigenfunctions $\f_n(0)$ form an orthonormal basis. Of course this is no longer the case for the non-selfadjoint operator $\Ta$ with $a \neq 0$. However we can prove that in this case we have a Riesz basis.\\

We recall that the sequence $\seq \f n$ of vectors in the Hilbert space $\Hc$ is said to be a Riesz basis if there exists a bounded operator $\Th \in \Lc(\Hc)$ with bounded inverse and an orthonormal basis $\seq e n$ of $\Hc$ such that $\f_n = \Th e_n$ for all $n \in \N$. 
In this case any $f \in \Hc$ can be written as $\sum_{n\in\N} f_n \f_n$ with $\seq f n \in l^2(\C)$, and there exists $C \geq 1$ such that for all $f = \sum_{n\in\N} f_n \f_n\in \Hc$ we have 
\[
C\inv \sum_{n\in\N} \abs{f_n}^2 \leq \nr{\sum_{n\in\N} f_n \f_n}_\Hc^2 \leq C \sum_{n\in\N} \abs{f_n}^2.
\]
In these estimates we can take $C = \max \left( \nr{\Th}^2 , \nr {\Th\inv}^2 \right)$.

Let $\seq \f n$ be a Riesz basis of $\Hc$ with $\Th$ and $\seq e n$ as above. If we set $\f^* _n = (\Th\inv)^* e_n$ for all $n \in \N$ then $(\f^*_n)_{n\in\N}$ is also a Riesz basis, called the dual basis of $\seq \f n$. In particular for all $j,k \in \N$ we have 
\begin{equation} \label{eq-dual-bases}
\innp{\f_j}{\f^*_k}_\Hc = \d_{j,k}.
\end{equation}
We refer to \cite{christensen} for more details about Riesz bases. Now we want to prove that for any $a \in \R$ the sequence of eigenfunctions for the operator $\Ta$ is a Riesz basis of $L^2(0,\oo)$.

\begin{proposition} \label{prop-riesz}
For all $a \in \R$ the sequence $(\f_n(a))_{n\in\N}$ defined in \eqref{def-phi} is a Riesz basis of $L^2(0,\oo)$. Moreover for all $R > 0$ there exists $C \geq 0$ such that for $a \in[-R,R]$ and $\seq c n \in l^2(\C)$ we have
\[
C\inv \sum_{n = 0}^{\infty} \abs{c_n}^2 \leq \nr{\sum_{n = 0}^{\infty} c_n \f_n(a)}^2 \leq C \sum_{n = 0}^{\infty} \abs{c_n}^2.
\]
If $C$ was chosen large enough and if $\seq c n \in l^2(\C)$ is such that $\sum_{n=0}^\infty \abs{\srev_n(a) c_n}^2 < \infty$ we also have
\[
C\inv \sum_{n = 0}^{\infty} \abs{\srev_n(a) c_n}^2 \leq \nr{\sum_{n = 0}^{\infty} c_n \f'_n(a)}^2 \leq C \sum_{n = 0}^{\infty} \abs{\srev_n(a) c_n}^2.
\]
\end{proposition}

For similar results we refer to \cite{mikhajlov62} (see also Lemma XIX.3.10 in \cite{dunford-schwartz}). For the proof we need the following lemma:

\begin{lemma} \label{lem-phin}
Let $R > 0$. Then there exists $C \geq 0$ such that for $a \in [-R,R]$ and $j,k \in \N$ with $j < k$ we have 
\[
\innp{\f_j(a)}{\f_k(a)}_{L^2(0,\oo)} \leq \frac C {\pppg j (k-j)} \quad \text{and} \quad \innp{\f'_j(a)}{\f'_k(a)}_{L^2(0,\oo)} \leq C \frac {k}{k-j}.
\]
\end{lemma}

\begin{proof}
Let $e_n(a,x) = e^{i \srev_n (a) x}$ and $\tilde e_n(a,x) = e^{-i \srev_n (a) x}$. According to Proposition \ref{prop-srev} we have
$
\l_n(a) = n\n + O(n\inv)
$
(here and below all the estimates are uniform in $a \in [-R,R]$), so
\begin{equation*}
\begin{aligned}
\nr{e_n(a)}^2_{L^2(0,\oo)}
= \frac {e^{-2\oo \Im(\srev_n(a))} - 1 }{-2 \Im(\srev_n(a))} = \oo + O(n\inv).
\end{aligned}
\end{equation*}
Similarly $\nr{\tilde e_n(a)}^2_{L^2(0,\oo)} = \oo + O(n\inv)$, and
$\innp{e_n(a)}{\tilde e_n(a)} =O(n\inv)$. Moreover, with \eqref{def-phi} we see that 
\[
A_n(a) = \frac 1 {\sqrt {2\oo}} + O(n\inv).
\]
Now let $j,k \in \N$ with $j < k$. We omit the dependence in $a$ for $\f_j$, $e_j$, $\tilde e_j$ and $\l_j$. We have 
\[
\innp{e_j}{e_k} = \frac {e^{i(\srev_j - \bar {\srev _k}) \oo} -1} {\srev_j - \bar {\srev _k}} 
\quad \text{and} \quad 
\innp{\tilde e_j}{\tilde e_k} = \frac {e^{-i(\srev_j - \bar {\srev _k}) \oo} -1} {-(\srev_j - \bar {\srev _k})} .
\]
Since
\[
\srev_j - \bar {\srev_k} = (j-k) \n + O(1/j)   
\]
we have
\[
\abs{\innp{e_j}{e_k}} + \abs {\innp{\tilde e_j}{\tilde e_k}} \lesssim \frac 1 {k-j} 
\quad \text{and} \quad 
\abs{\innp{e_j}{e_k}  + \innp{\tilde e_j}{\tilde e_k}} \lesssim \frac 1 {\pppg j (k-j)} .
\]
Similarly
\[
\abs{\innp{\tilde e_j}{e_k}} + \abs {\innp{ e_j}{\tilde e_k}} \lesssim \frac 1 {k+j} 
\quad \text{and} \quad 
\abs{\innp{\tilde e_j}{e_k}  + \innp{e_j}{\tilde e_k}} \lesssim \frac 1 {\pppg j (k+j)} .
\]
And finally
\begin{align*}
\abs{\innp{\f_j}{\f_k}}
& = \abs{A_j\bar {A_k}\innp{  e_j + \frac {\srev_j - a }{\srev_j + a } \tilde e_j } {  e_k + \frac {\srev_k - a }{\srev_k + a } \tilde e_k}}\\
& \lesssim \abs{ \innp{ e_j}{ e_k} +\innp{ e_j}{\tilde e_k}+\innp{\tilde e_j}{ e_k}+\innp{\tilde e_j} {\tilde e_k} }\\
&  \quad + \pppg j \inv \left( \abs{\innp{ e_j}{ e_k}} + \abs{\innp{ e_j}{\tilde e_k}} + \abs{\innp{\tilde e_j}{ e_k}} + \abs{\innp{\tilde e_j} {\tilde e_k}} \right)\\
& \lesssim \frac 1 {\pppg j (k-j)}. 
\end{align*}
The second estimate is proved similarly, using again that $\abs{\srev_n(a)} = n\n + O(n\inv)$ for large $n$.
\end{proof}

\begin{proof} [Proof of Proposition \ref{prop-riesz}]
\stepp Let $a \in \R$ and $\seq c n \in l^2(\C)$. Let $C\geq 0$ be given by Lemma \ref{lem-phin}. For $N,p \in \N$ we have
\begin{align*}
\nr{\sum_{n = N}^{N+p} c_n \f_n(a)}^2 - \sum_{n=N }^{N+p} \abs{c_n}^2
&=  \sum_{j=N}^{N+p} \sum_{k=j+1}^{N+p} 2 \Re \left( c_j \bar {c_k} \innp{\f_j}{\f_k} \right) \\
%& \leq  2C \sum_{j=0}^\infty \sum_{k=j+1}^\infty \frac {\abs{c_j} \abs{c_k}}{\pppg j(k-j)} \\
& \leq  2C \sum_{j=0}^\infty \frac {\abs{c_j}}{\pppg j} \sum_{k=1}^\infty \frac {\abs{c_{k+j}}}{k} \\
& \lesssim \nr{c}_{l^2(\C)}^2.
\end{align*}
This proves that the series $\sum_{n = 0}^\infty c_n \f_n(a)$ converges in $L^2(0,\oo)$ and
\[
 \nr{\sum_{n = 0}^\infty c_n \f_n(a)}_{L^2(0,\oo)}^2 \lesssim \sum_{n=0}^\infty \abs{c_n}^2.
\]

\stepp Let $n,m \in \N$ be such that $n\neq m$. In $L^2(0,\oo)$ we have
\begin{align*}
\srev_n(a)^2 \innp{\f_n(a)}{\f_m(-a)}
& = \innp{\Ta \f_n(a)}{\f_m(-a)} = \innp{\f_n(a)}{\Ta^* \f_m(-a)}\\
& = {\srev_m(a)^2} \innp{\f_n(a)}{\f_m(-a)},
\end{align*}
and hence $\innp{\f_n(a)}{\f_m(-a)} = 0$. On the other hand, with \eqref{def-phi} we can check that for $n$ large enough we have $\innp{\f_n(a)}{\f_n(-a)} \neq 0$.
Now let $\seq c n \in l^2(\C)$ be such that $\sum_{n=0}^\infty c_n \f_n(a) = 0$. Taking the inner product with $\f_m(-a)$ we see that $c_m = 0$ for $m > N$ if $N$ is chosen large enough. Then for all $k \in \Ii 0 N$ we have 
\[
0 = \Ta^k \sum_{n=0}^N c_n \f_n(a) = \sum_{n=0}^N \srev_n(a)^{2k} c_n \f_n(a).
\]
Since the eigenvalues $\srev_n(a)^2$ for $n\in\Ii 0 N$ are pairwise disjoint, this proves that for all $n \in \Ii 0 N$ we have $c_n \f_n(a) = 0$, and hence $c_n=0$. Finally the map $(c_n) \in l^2(\C) \mapsto \sum c_n\f_n(a) \in L^2(0,\oo)$ is one-to-one.

\stepp As above we can check that for $N\in\N$ we have
\begin{align*}
\nr{\sum_{n = N}^{\infty} c_n \f_n(a)}^2 - \sum_{n=N}^{\infty} \abs{c_n}^2
\lesssim \sum_{n=N}^{\infty} \abs{c_n}^2  \sqrt{\sum_{n=N}^\infty \frac {1}{\pppg n^2}}.
\end{align*}
This is less that $\frac 12 \sum_{n=N}^{\infty} \abs{c_n}^2$ if $N$ is chosen large enough. Let such an integer $N$ be fixed. 
Now assume by contradiction that the sequences $\seq a m \in [-R,R]^\N$ and $(c^m)_{m\in \N}$ in $l^2(\C)^\N$ are such that $\nr{c^m}_{l^2} = 1$ for all $m \in \N$ and
\[
\nr{\sum_{n=0}^\infty c_n^m \f_n(a_m)}^2 \limt m \infty 0.
\]
For $m \in \N$ we set $f_m = \sum_{n=0}^{N-1} c_n^m \f_n(a_m)$ and $g_m = \sum_{n=N}^{\infty} c_n^m \f_n(a_m)$. After extracting a subsequence if necessary we can assume that $a_m$ converges to some $a \in [-R,R]$ and $f_m$ converges to some $f \in L^2(0,\oo)$. Let $P_m$ (respectively $P$) denote the orthogonal projection on $\operatorname{span}(\f_j(a_m))_{j\geq N}$ (respectively on $\operatorname{span}(\f_n(a))_{n\geq N}$). We have
\[
\nr{f_m + g_m}^2 = \nr{f_m}^2 + 2 \Re \innp{P_m f_m }{g_m} + \nr{g_m}^2 \geq \nr{f_m}^2 - \nr{P_m f_m}^2 \limt m \infty \nr f ^2 - \nr{P f}^2.
\]
This gives $f = P f$, and hence $f = 0$. 
This gives a contradiction with 
\[
0 = \lim _{m \to \infty} \nr{g_m}^2 \geq  \lim _{m\to \infty} \frac 12 \sum_{n=N}^\infty \abs{c_n^m}^2 = \frac 12
\]
and proves the first inequality of the proposition.

\stepp It remains to prove that the sequence $(\f_n(a))$ is complete. We know that the family $\big( \f_n(0) \big)_{n\in\N}$ is an orthonormal basis of $L^2(0,\oo)$. Since
$
\f_n(a) =  \f_n(0) +  O \big( n\inv \big)
$
in $L^2(0,\oo)$, this follows from a perturbation argument (see Theorem V.2.20 in \cite{kato}). This concludes the proof.
\end{proof}

For $a \in \R$ we denote by $(\f_n^*(a))_{n\in\N}$ the dual basis of $(\f_n(a))_{n\in\N}$. We have $\f_n^*(a) = \f_n(-a) = \bar{\f_n(a)}$.

\section{Separation of variables - Spectrum of the model operator} \label{sec-separation}

In this section we use the results on the transversal operator $\Ta$ to prove spectral properties for the full operator $\Ha$ when $a$ is constant on $\partial \O$. Most of the results of this section are inspired by the $\Pc \Tc$-symmetric analogs (see \cite{borisovk08}).
Let $a > 0$ be fixed. We set
\[
\mathfrak{S}_a = \bigcup _{n \in \N} \{\l_n(a)^2\} + \R_+  = \singl{\l_n(a)^2 + r, n \in \N, r \in \R_+} \subset \C.
\]

\begin{proposition} 
We have $\mathfrak S _a \subset \s(\Ha)$.
\end{proposition}

\begin{proof}
Let $n \in \N$, $r \geq 0$ and $z = \srev_n(a)^2 + r \in \mathfrak{S}_a$. Let $\seq f m$ be a sequence in $H^2(\R^{\dimm-1})$ such that $\nr{f_m}_{L^2(\R^{\dimm-1})} = 1$ for all $m \in \N$ and $\nr{(-\D_x -r)f_m}_{L^2(\R^{\dimm-1})} \to 0$ as $m \to \infty$. For $m \in \N$ and $(x,y) \in \R^{\dimm-1} \times ]0,\oo[$ we set $u_m(x,y) = f_m (x) \f_n(a;y)$. Then $u_m \in \Dom(\Ha)$ and $\nr{u_m}_{L^2(\O)} = 1$ for all $m \in \N$. Moreover, according to \eqref{separation-variable-Ha} we have
\[
\nr{(\Ha - z)u_m}_{L^2(\O)} = \nr{(-\D_x - r) f_m}_{L^2(\R^{\dimm-1})} \limt m \infty 0.
\]
This implies that $z \in \s(\Ha)$.
\end{proof}

For $u \in L^2(\O)$, $n\in\N$ and $x \in \R^{\dimm-1}$ we set $u_n (x) = \innp{u(x, \cdot)}{\f^*_n(a)}_{L^2(0,\oo)} \in \C$. This gives a sequence of functions $u_n$ defined almost everywhere on $\R^{\dimm-1}$.

\begin{proposition} \label{prop-separation-variables}
Let $u \in L^2(\O)$. Then $u_n \in L^2(\R^{\dimm-1})$ for all $n \in \N$ and on $L^2(\O)$ we have
\[
u= \sum_{n\in\N} u_n \otimes \f_n(a).
\]
\end{proposition}

% This is Lemma 4.1 in \cite{borisovk08}.

\begin{proof}
For $N \in \N$ we set $v_N = \sum_{n = 0}^N u_n \otimes \f_n(a)$. For almost all $x \in \R^{\dimm-1}$, $v_N(x)$ defines a function in $L^2(0,\oo)$ which goes to $u(x)$ as $N \to \infty$. According to Proposition \ref{prop-riesz} we have
\begin{align*}
\nr{v_N(x)}^2_{L^2(0,\oo)} \lesssim   \sum_{n=0}^N \abs{u_n(x)}^2 \lesssim  \sum_{n=0}^\infty \abs{u_n(x)}^2 \lesssim   \nr{u(x)}^2_{L^2(0,\oo)},
\end{align*}
where all the estimates are uniform in $N$. Since the map $x \mapsto  \nr{u(x)}^2_{L^2(0,\oo)}$ belongs to $L^1(\R^{\dimm-1})$, we can apply the dominated convergence theorem to conclude that $u-v_N \to 0$ in $L^2(\O)$.
\end{proof}

\begin{proposition} \label{prop-dec-res}
\begin{enumerate}[(i)]
\item Let $u \in \Dom(\Ha)$. Then $u_n \in H^2(\R^{\dimm-1})$ for all $n \in \N$ and in $L^2(\O)$ we have
\[
\Ha u = \sum_{n\in\N} \big( - \D_x + \srev_n(a)^2 \big) u_n \otimes \f_n(a).
\]
\item Let $z \in \C \setminus \mathfrak{S}_a$. Then $z$ belongs to the resolvent set of $\Ha$ and for all $u \in L^2 (\O)$ we have
\[
(\Ha -z)\inv u = \sum_{n\in\N} (-\D_x + \srev_n(a)^2-z)\inv u_n \otimes \f_n(a). 
\]
In particular there exists $C \geq 0$ such that for all $z \in \C \setminus \mathfrak{S}_a$ and $u \in L^2(\O)$ we have
\[
\nr{(\Ha -z)\inv u}_{L^2(\O)} \leq \frac C {d(z,\mathfrak{S}_a)} \nr{u}_{L^2(\O)}.
\]
\end{enumerate}

\end{proposition}

\newcommand{\Uj}{\tilde R_n(u)}

\begin{proof}
\stepp Let $z \in \C \setminus \mathfrak{S}_a$ and $u \in L^2(\O)$. Let $u_n \in L^2(\R^{\dimm-1})$, $n\in\N$, be as above. For $n \in \N$ we set $\Uj  = (-\D_x + \srev_n(a)^2 - z)\inv u_n \in H^2 (\R^{\dimm-1})$ and $R_n(u)= \Uj \otimes \f_n(a) \in \Dom(\Ha)$. Using the standard spectral properties of the self-adjoint operator $-\D_x$, we see that on $L^2(\R^{\dimm-1})$ we have
\begin{equation} \label{bk45}
\nr{\Uj} \lesssim \frac {\nr{u_n}} {d(z, \srev_n(a)^2 +\R_+ )} \lesssim_z  \frac {\nr{u_n}}{\pppg n ^{2}} \quad \text{and} \quad \nr{\partial _x \Uj} \lesssim_z \frac {\nr{u_n}}{\pppg n }.
\end{equation}
The first estimate is uniform in $z$ but not the others. With Proposition \ref{prop-riesz} we obtain for $N,p \in \N$
\[
\nr{\sum_{n=N}^{N+p} R_n(u)}^2_{H^1(\O)} \lesssim_z \sum_{n=N}^{N+p} \left( \nr{\Uj}^2_{H^1(\R^{\dimm-1})} + \pppg n^2 \nr{\Uj}^2_{L^2(\R^{\dimm-1})} \right) \limt N \infty 0.
\]
This proves that the series $\sum R_n(u)$ converges to some $R(u) \in H^1(\O)$. Moreover, with the first inequality of \eqref{bk45} and Proposition \ref{prop-riesz} again, we see that 
\begin{equation} \label{estim-R}
 \nr{R(u)}_{L^2(\O)} \lesssim \frac {\nr{u}_{L^2(\O)}} {d (z ,\mathfrak S_a)},
\end{equation}
uniformly in $z$. It remains to see that 
\[
\forall \vf \in H^1(\O), \quad q_a(R(u),\vf) - z \innp {R(u)} \vf = \innp u \vf,
\]
which comes from the fact that this is true with $R(u)$ and $u$ replaced by $R_n(u)$ and $u_n \otimes \f_n(a)$ respectively. This proves that $R(u) \in \Dom(\Ha)$ and $(\Ha-z) R(u) = u$. If $\Re(z) < 0$, we already know that $(\Ha -z)$ has a bounded inverse on $L^2(\O)$, and hence we have $(\Ha -z)\inv = R$. This proves the second statement of the proposition when $\Re(z) < 0$.

\stepp Let $ u \in \Dom(\Ha)$ and $v = (\Ha+1) u \in L^2(\O)$. For $n \in \N$ and almost all $x \in \R^{\dimm-1}$ we set $v_n(x) = \innp{v(x,\cdot)}{\f_n^*(a)}_{L^2(0,\oo)}$. According to (ii) applied with $z = -1$ we have 
\[
u = (\Ha +1)\inv v = \sum_{n\in\N} \big(-\D_x + \l_n(a)^2 +1\big)\inv v_n \otimes \f_n(a).
\] 
By uniqueness for the decomposition of $u(x,\cdot)$ with respect to the Riesz basis $(\f_n(a))_{n\in\N}$, we have for all $n \in \N$
\[
u_n = \big(-\D_x + \l_n(a)^2 +1\big)\inv v_n.
\]
This proves that $u_n \in H^2(\R^{\dimm-1})$. Then 
\begin{align*}
\Ha u = v - u
& = \sum_{n \in \N}  \left( 1 - \big(-\D_x + \l_n(a)^2 +1\big)\inv\right) v_n \otimes \f_n(a)\\
& = \sum_{n \in \N}  \big(-\D_x + \l_n(a)^2 \big) u_n \otimes \f_n(a).
\end{align*}
This proves the first statement of the proposition.

\stepp It remains to finish the proof of (ii). Let $z \in \C \setminus \mathfrak{S}_a$ and $w \in \Dom(\Ha)$. With (i) we see that $R\big((\Ha-z)w\big) = w$. Since we already know that $(\Ha-z)R(u) = u$ for all $u \in L^2(\O)$, this proves that $R$ is a bounded inverse for $(\Ha -z)$ on $L^2(\O)$. The estimate on $(\Ha -z)\inv$ follows from \eqref{estim-R}, and the proposition is proved.
\end{proof}

As a first application of this proposition, we can check that the operator $\Ha$ has no eigenvalue, as is the case for $-\D_x$:

\begin{corollary} \label{cor-no-vp}
The operator $\Ha$ has no eigenvalue.
\end{corollary}

\begin{proof}
Let $z \in \C$ and $u  \in \Dom(\Ha)$ be such that $(\Ha-z) u =0$. Then according to the first item of Proposition \ref{prop-dec-res} we have
\[
\sum_{n \in \N} \left( -\D + \srev_n(a)^2 - z \right) u_n \otimes \f_n(a) = 0.
\]
This implies that for all $n \in \N$ we have 
\[
 \left( -\D + \srev_n(a)^2 - z \right) u_n = 0
\]
in $L^2(\R^{\dimm-1})$, and hence $u_n = 0$ since the operator $-\D_x$ has no eigenvalue. Finally $u = 0$, and the proposition is proved.
\end{proof}

However, the main point in Proposition \ref{prop-dec-res} is the second:

\begin{corollary} \label{cor-gap-spectral-aconstant}
Theorem \ref{th-gap-spectral} holds when $a > 0$ is constant.
\end{corollary}

\begin{proof}
According to the last statement of Proposition \ref{prop-srev} and the fact that $\Im(\srev_n(a)^2) < 0$ for all $n \in \N$, there exists $\tilde \g > 0$ such that for all $z \in \mathfrak{S}_a$ we have $\Im(z) \leq - \tilde \g$. With Proposition \ref{prop-dec-res}, the conclusion follows.
\end{proof}

\begin{remark} \label{rem-a-diff-constantes}
To simplify the proof we have only considered the case where $a$ is equal to the same constant on both sides of the boundary. However we can similarly consider the case where $a$ is equal to some positive constant $\ao$ on $\R^{\dimm-1} \times \singl 0$ and to another positive constant $\al$ on $\R^{\dimm-1} \times \singl l$. One of these two constants may even be zero. We refer to Section \ref{sec-non-diss} for this kind of computations.
\end{remark}

\section{Non-constant absorption index} \label{sec-a-variable}

\newcommand{\Hao}{H_{a_0}}
\newcommand{\tHa}{\tilde H_{a}}
\newcommand{\tHao}{\tilde H_{a_0}}

In this section we prove Theorems \ref{th-gap-spectral} and \ref{th-smoothing} for a non-constant absorption index $a$. For $b \in W^{1,\infty}(\partial \O)$ we denote by $\Th_b \in \Lc(H^1(\O) , H\inv(\O))$ the operator such that for all $\f ,\p \in H^1(\O)$ 
\[
\innp{\Th_b \f}\p _{H\inv(\O),H^1(\O)} = \int_{\partial \O} b\f \bar \p .
\]
We denote by $\th_b$ the corresponding quadratic form on $H^1(\O)$. 
We also denote by $\tHa$ the operator in $\Lc \big(H^1(\O),H\inv(\O) \big)$ such that $\innp{\tHa \f}\p _{H\inv,H^1} = q_a(\f,\p)$ for all $\f,\p \in H^1(\O)$. Let $z \in \C_+$. According to the Lax-Milgram Theorem, $(1+i)(\tHa-z)$ is an isomorphism from $H^1(\O)$ to $H\inv(\O)$. Moreover, for $f \in L^2(\O) \subset H\inv(\O)$ we have 
\[
(\tHa -z)\inv f = (\Ha -z)\inv f.
\]

The following proposition relies on a suitable version of the so-called quadratic estimates:

\begin{proposition} \label{prop-quadratic-estimates}
Let $a_0 > 0$ be as in the statement of Theorem \ref{th-energy-decay}. Assume that \eqref{hyp-a0-a-a1} holds everywhere on $\partial\O$. Let $B \in \Lc\big(H^1(\O),L^2(\O) \big)$. Then there exists $C \geq 0$ such that for all $z \in \C_+$ we have 
\[
\nr{B (\tHa-z)\inv B^*}_{\Lc(L^2(\O))} \leq C \nr{B (\tHao-z)\inv B^*}_{\Lc(L^2(\O))}.
\]
\end{proposition}

\begin{proof}
For $z \in \C_+$ the resolvent identity applied to $\tHa = \tHao + \Th_{a-a_0}$ gives 
\begin{equation} \label{eq-res-identity}
B(\tHa -z)\inv B^* = B (\tHao -z)\inv B^* - B (\tHao -z)\inv \Th_{a-a_0} (\tHa -z)\inv B^*.
\end{equation}
Let $\f,\p \in L^2(\O)$. Since $\Th_{a-a_0}$ is associated to a non-negative quadratic form on $H^1(\O)$, the Cauchy-Schwarz inequality gives 
\begin{eqnarray*}
\lefteqn{\innp{B(\tHao -z)\inv \Th_{a-a_0} (\tHa -z)\inv B^*\f} \p _{L^2}}\\
&& = \innp{ \Th_{a-a_0} (\tHa -z)\inv B^* \f} {(\tHao^* - \bar z)\inv B^* \p} _{H\inv,H^1}\\
&& \leq \th_{a-a_0} \big((\tHa -z)\inv B^* \f \big)^{\frac 12}  \times \th_{a-a_0} \big((\tHao^* -\bar z)\inv B^* \p\big) ^{\frac 12}.
\end{eqnarray*}
The first factor is estimated as follows:
\begin{eqnarray*}
\lefteqn{\th_{a-a_0} \big((\tHa -z)\inv B^* \f \big)\leq \th_{a} \big((\tHa -z)\inv B^* \f \big)}\\
&& \leq \frac 1 {2i} \innp{ 2i (\Th_{a} + \Im z) (\tHa -z)\inv B^* \f} {(\tHa - z)\inv B^* \f} _{H\inv,H^1}\\
&& \leq \frac 1 {2i} \innp{ B \big((\tHa -z)\inv-(\tHa^* - \bar z)\inv \big)  B^*\f} {\f} _{L^2}\\
&& \leq \nr{B(\tHa-z)\inv B^*}_{\Lc(L^2)} \nr{\f}_{L^2}^2.
\end{eqnarray*}
We can proceed similarly for the other factor, using the fact that $a-a_0 \leq \a a_0$ for some $\a \geq 0$. Thus there exists $C \geq 0$ such that 
\begin{align*}
\nr{B (\tHao -z)\inv \Th_{a-a_0} (\tHa -z)\inv B^* }
& \leq C \nr{B (\tHao -z)\inv B^*}^{\frac 12} \nr{B (\tHa -z)\inv B^*}^{\frac 12}\\
& \leq \frac {C^2}2 \nr{B (\tHao -z)\inv B^*} + \frac 12 \nr{B (\tHa -z)\inv B^*}.
\end{align*}
With \eqref{eq-res-identity}, the conclusion follows.
\end{proof}

Now we can finish the proof of Theorem \ref{th-gap-spectral}:

\begin{proof}[Proof of Theorem \ref{th-gap-spectral}]
According to Corollary \ref{cor-gap-spectral-aconstant} and Proposition \ref{prop-quadratic-estimates} applied with $B = \Id_{L^2(\O)}$ (note that we can simply replace $\tHa$ by $\Ha$ when $B \in \Lc(L^2(\O))$), there exists $C > 0$ such that $\nr{(\Ha-z)\inv}_{\Lc(L^2(\O))} \leq C$ for all $z \in \C_+$. Then it only remains to choose $\tilde \g \in \big] 0, \frac 1 {C} \big[$ to conclude. The second statement, concerning the case where $a$ vanishes on one side of the boundary, is proved similarly.
\end{proof}

Let us now turn to the proof of Theorem \ref{th-smoothing}. We first prove another resolvent estimate in which we see the smoothing effect in weighted spaces:

\begin{proposition} \label{prop-weighted}
Let $\delta > \frac 12$. Then there exists $C \geq 0$ such that for all $z \in \mathbb C_+$ we have 
\[
\nr{\left< x \right> ^{-\d} (1- \Delta_x)^{\frac 14} (H_a-z)^{-1} (1- \Delta_x)^{\frac 14} \left< x \right> ^{-\delta}}_{\mathcal L (L^2(\Omega))} \leq C.
\]
\end{proposition}

\begin{proof}
It is known for the free laplacian that for $\left| \mathop{\rm{Re}} (\zeta) \right| \gg 1$ and $\mathop{\rm{Im}}(\zeta) > 0$ we have 
\[
\left\Vert \left< x \right>^{-\delta} (-\Delta_x -\zeta)^{-1} \left< x \right>^{-\delta}\right\Vert _{\mathcal {L}(L^2(\mathbb {R}^{d-1}))} \lesssim \left< \z \right > ^{-\frac 12}
\]
and hence 
\[
\left\Vert\left< x \right>^{-\d} (1- \D_x)^{\frac 14}(-\Delta_x -\zeta)^{-1} (1- \Delta_x)^{\frac 14}\left< x \right>^{-\delta} \right \Vert_{\mathcal L (L^2(\mathbb R^{d-1}))} \lesssim 1.
\]
Thus if $a$ is constant we have 
\[
\left\Vert\left< x \right>^{-\d} (1- \D_x)^{\frac 14}(-\Delta_x + \lambda_n(a)^2 -z)^{-1} (1- \Delta_x)^{\frac 14}\left< x \right>^{-\d}\right \Vert_{\mathcal {L}(L^2(\mathbb {R}^{d-1})) } \lesssim 1,
\]
uniformly in $z \in \C_+$ and $n \in \mathbb {N}$. In this case we obtain the result using the separation of variables as in Section \ref{sec-separation}. Then we conclude with Proposition \ref{prop-quadratic-estimates} applied with $B = \left< x \right> ^{-\delta} (1-\Delta_x)^{\frac 14}$. In fact, we first obtain an estimate on the resolvent $(\tilde H_a -z)^{-1}$, but this proves that the operator $\left< x \right> ^{-\d} (1- \Delta_x)^{\frac 14} (H_a-z)^{-1} (1- \Delta_x)^{\frac 14} \left< x \right> ^{-\delta}$ extends to a bounded operator on $L^2(\Omega)$, and then the same estimate holds for the corresponding closure.
\end{proof}

\begin{proof} [Proof of Theorem \ref{th-smoothing}]
With the second estimate of Proposition \ref{prop-weighted}, we can apply the theory of relatively smooth operators (see \S XIII.7 in \cite{rs4}). However, since $\Ha$ is not self-adjoint but only maximal dissipative, we have to use a self-adjoint dilation (see \cite{nagyf}) of $\Ha$, as is done in the proof of \cite[Proposition 5.6]{art-mourre} (see also Proposition 2.24 in \cite{these}).
\end{proof}

\section{Time decay for the Schr\"odinger equation}  \label{sec-time-decay}

In this section we prove Theorem \ref{th-energy-decay}.
Let $u_0 \in \Dom(\Ha)$ and let $u$ be the solution of the problem \eqref{schrodinger}. We know that $\nr{u(t)}_{L^2(\O)} \leq \nr{u_0}_{L^2(\O)}$ for all $t \geq 0$, so the result only concerns large times.
Let $\tilde \g > 0$ be given by Theorem \ref{th-gap-spectral} and $\g = \tilde \g / 3$.
%such that the resolvent $(\Ha-z)\inv$ is well-defined and uniformly bounded on $\C_{4\g}$. 
Let $\h \in C^\infty(\R)$ be equal to 0 on $]-\infty,1]$ and equal to 1 on $[2,+\infty[$. For $t \in \R$ we set 
\[
u_\h (t) =  \h(t) u(t),
\]
and for $z \in \C_+$:
\[
v(z) = \int_\R e^{itz} u_\h(t) \, dt.
\]
The map $t \mapsto e^{-\g t}u_\h(t)$ belongs to $L^1(\R) \cap L^2(\R) \cap C^1(\R)$ and its derivative is in $L^1(\R)$ so $\t \mapsto v(\t+i\g)$ is bounded and decays at least like $\pppg \t \inv$. In particular it is in $L^2(\R)$.
For $R > 0$ we set 
\[
u_R(t) = \frac 1 {2\pi} \int_{-R}^R e^{-it(\t+i\g)} v(\t+i\g) \, d\t.
\]
Then
\[
\nr{e^{-t\g} (u_\h- u_R)}_{L^2(\R_t,L^2(\O))} \limt R {+\infty} 0.
\]
Since $u_\h$ is continuous, Theorem \ref{th-energy-decay} will be proved if we can show that there exists $C \geq 0$  which does not depend on $u_0$ and such that for all $t \geq 0$ we have 
\begin{equation} \label{estim-uR}
\limsup_{R \to \infty} \nr{u_R(t)}_{L^2(\O)} \leq C e^{-\g t} \nr{u_0}_{L^2(\O)}.
\end{equation}
For $z \in \C$ we set 
\[
\th(z) = -i \int_\R e^{itz} \h'(t) u(t) \, dt = -i \int_1^2 e^{itz} \h'(t) u(t) \, dt.
\]
Let $z \in \C_+$. We multiply \eqref{schrodinger} by $\h(t) e^{itz}$ and integrate over $t \in \R$. After partial integration we obtain
\begin{equation*} %\label{def-v}
 v(z) = (\Ha -z)\inv \th(z).
\end{equation*}
Then $v$ extends to a holomorphic function on $\C_{3\g}$, taking this equality as a definition.
According to the Cauchy Theorem we have in $L^2(\R_t)$
\begin{equation} \label{estim-uR-w}
\begin{aligned}
\lim_{R \to \infty} u_R(t)
%& 
= \frac 1 {2\pi} \lim_{R \to \infty} \int_{-R}^R e^{-it(\t-2i\g)} v(\t-2i\g)\,d\t=  {e^{-2\g t}}  \lim_{R \to \infty} \widetilde {u_R}(t),
\end{aligned}
\end{equation}
where for $t \in \R$ we have set 
\[
\widetilde {u_R}(t) =  \int_{-R}^R e^{-it\t}v(\t-2i\g) \, d\t.
\]
According to Plancherel's equality and Theorem \ref{th-gap-spectral} we have uniformly in $R >0$:
\begin{align*}
\int_\R \nr{\widetilde {u_R} (t)}^2_{L^2(\O)} \, dt 
& = \int_{-R}^R \nr{\big(\Ha -(\t-2i\g)\big)\inv \th(\t-2i\g)}_{L^2(\O)}^2 \, d\t \\
& \lesssim \int_\R \nr{\th(\t-2i\g)}_{L^2(\O)}^2 \, d\t \\
& \lesssim \int_\R e^{2\g t}\abs{\h'(t)} \nr{u(t)}_{L^2(\O)}^2 \, dt \\
& \lesssim \nr{u_0}^2_{L^2(\O)}.
\end{align*}
In particular there exists $C \geq 0$ such that for $u_0 \in \Dom(\Ha)$ and $R > 0$ we can find $T(u_0,R) \in [0,1]$ which satisfies
\[
\nr{\widetilde {u_R} (T(u_0,R))}_{L^2(\O)} \leq C \nr{u_0}_{L^2(\O)}.
\]
Let $R > 0$. Then $\widetilde {u_R} \in C^1(\R)$ and for $t \geq 1$ we have
\[
\widetilde {u_R} (t) = e^{-i(t-T(u_0,R))\Ha} \widetilde {u_R} (T(u_0,R)) +  \int_{T(u_0,R)} ^t \frac {\partial}{\partial s} \left( e^{-i(t-s)\Ha} \widetilde {u_R}(s) \right) \,ds,
\]
where
\begin{align*}
\frac {\partial}{\partial s} \left( e^{-i(t-s)\Ha} \widetilde {u_R}(s) \right)
& = \frac {\partial}{\partial s}\int_{-R}^R e^{-is\t} e^{-i(t-s)\Ha} \big(\Ha -(\t-2i\g)\big)\inv \th(\t-2i\g) \, d\t\\
& = i \int_{-R}^R e^{-i(t-s)\Ha} e^{-is\t} (\Ha-\t) \big( \Ha - (\t-2i\g) \big) \inv \th(\t-2i\g) \, d\t\\
& = 2\g e^{-i(t-s)\Ha} \widetilde {u_R} (s) + i e^{-i(t-s)\Ha}  \int_{-R}^R e^{-is\t}  \th(\t-2i\g) \, d\t.
\end{align*}
This proves that the map $s \mapsto \frac {\partial}{\partial s} \left( e^{-i(t-s)\Ha} \widetilde {u_R}(s) \right)$ belongs to $L^2([0,t],L^2(\O))$ uniformly in $t$ and $R>0$, and its $L^2([0,t],L^2(\O))$ norm is controlled by the norm of $u_0$ in $L^2(\O)$. We finally obtain $C \geq 0$ such that for all $t \in \R$ and $R > 0$ we have
\[
\nr{\widetilde {u_R} (t)}_{L^2(\O)} \leq C \pppg t^{\frac 12} \nr {u_0}_{L^2(\O)}.
\]
With \eqref{estim-uR-w} this proves \eqref{estim-uR} and concludes the proof of Theorem \ref{th-energy-decay}.

\section{The case of a weakly dissipative boundary condition} \label{sec-non-diss}

\newcommand{\Tab} {T_{\al,\ao}}
\newcommand{\Hab} {H_{\al,\ao}}
\newcommand{\Tsab} {T_{s\al,s\ao}}
\newcommand{\lnab}{\srev_n(\al,\ao)}
\newcommand{\loab}{\srev_0(\al,\ao)}
\newcommand{\lnsab}{\srev_n(s\al,s\ao)}
\newcommand{\losab}{\srev_0(s\al,s\ao)}
\newcommand{\fab} {\f_n({\al,\ao})}

In this section we prove Theorem \ref{th-energy-decay-nondiss} about the problem \eqref{schrodinger-nondiss}. The absorption index $a$ now takes the value $\al$ on $\R^{\dimm-1} \times \singl l$ and $\ao$ on $\R^{\dimm-1} \times \singl 0$.

The proof follows the same lines as in the dissipative case, except that well-posedness of the problem is not an easy consequence of the general dissipative theory. We will use the separation of variables as in Section \ref{sec-separation} instead. Once we have a decomposition as in Proposition \ref{prop-separation-variables} for the initial datum, we can propagate each term by means of the unitary group generated by $-\D_x$ and define the solution of \eqref{schrodinger-nondiss} as a series of solutions on $\R^{\dimm-1}$.\\

Let us first look at the transverse problem. The transverse operator on $L^2(0,\oo)$ corresponding to the problem \eqref{schrodinger-nondiss} is now given by $\Tab = - \frac {d^2} {dy^2}$ with domain
\[
\Dom(\Tab) = \singl{u \in H^2(0,\oo) \st u'(0) = -i\ao u(0), u'(\oo) = i\al u(\oo)}.
\]
As already mentioned in Remark \ref{rem-a-diff-constantes}, we can reproduce exactly the same analysis as in Section \ref{sec-transverse} if $\al > 0$ and $\ao >0$ (or if one of them vanishes). In particular, there is no restriction on the sizes of these coefficients. The results we give here to handle the weakly dissipative case are also valid in this situation. \\

The strategy will be the same as in Section \ref{sec-transverse}, so we will only emphasize the differences. We first remark that 0 is an eigenvalue of $\Tab$ if and only if $\al=\ao=0$. Otherwise, $\srev ^2 \in \C^*$ is an eigenvalue of $\Tab$ if and only if 
\begin{equation} \label{eq-z-ab}
(\srev - \al)(\srev - \ao) e^{2i\srev \oo} = (\srev + \al) (\srev + \ao).
\end{equation}
We recover \eqref{eq-z} when $\al = \ao$.

\begin{lemma} \label{lemma-cross-nuN}
Let $\al,\ao \in \R$ and $\l \in \C^*$ be such that $\al+\ao \neq 0$ and $\l^2$ is an eigenvalue of $\Tab$. Then $\Re(\l) \notin  \n\N$.
\end{lemma}

We recall from \cite{krejcirikbz06} that if $\al+\ao = 0$ ($\Pc\Tc$-symmectric case) then $n^2 \n ^2 \in \s(\Tab)$ for all $n \in \N^*$ (see also Figure \ref{fig-30eigenvalues} for $\al+\ao > 0$ small).

\begin{proof}
\stepp We assume by contradiction that $\Re(\l) \in \n\N$. According to \eqref{eq-z-ab} we have 
\[
\frac {(\l + \al)(\l+\ao)}{(\l-\al)(\l-\ao)} = e^{2i\oo \l} = e^{-2 \oo \Im (\l)} \in \R_+^*.
\]
After multiplication by $\abs{\l - \al}^2 \abs{\l - \ao}^2 \in \R_+^*$ we obtain
\[
\big ( \abs \l ^2 - 2 i \al \Im (\l) - \al ^2 \big) \big ( \abs \l ^2 - 2 i \ao \Im (\l) - \ao ^2 \big)  \in \R_+^*.
\]
Taking the real and imaginary parts gives
\begin{equation} \label{cond-partie-reelle}
\abs{\l}^4 - (\al^2 + \ao^2) \abs \l ^2 + \al^2 \ao^2 - 4 \al \ao \Im(\l)^2 > 0
\end{equation}
and 
\begin{equation} \label{cond-partie-im}
2 \Im(\l) (\al +\ao) \big(\abs\l^2 - \al \ao \big) = 0.
\end{equation}

\stepp Assume that $\al \ao  \geq 0$. In this case $\Im (\l) \neq 0$ (for the same reason as in the proof of Proposition \ref{prop-realvp}), so \eqref{cond-partie-im} implies $\abs \l ^2 = \al \ao$. Then \eqref{cond-partie-reelle} reads
\[
- \al \ao (\al - \ao)^2 - 4\al \ao \Im(\l)^2 > 0,
\]
which gives a contradiction.

\stepp Now assume that $\al \ao < 0$. Then \eqref{cond-partie-im} implies $\Im(\l) = 0$ and hence $e^{2i\oo \l} = 1$. From \eqref{eq-z-ab} we now obtain 
\[
(\l-\al ) (\l-\ao) = (\l+\al ) (\l+\ao),
\]
which is impossible since $\l(\al + \ao) \neq 0$. This concludes the proof.
\end{proof}

\begin{proposition} \label{prop-Im-lab}
There exists $\rho > 0$ such that if $\abs {\al} + \abs {\ao} \leq \rho$ and $\al + \ao > 0$ then the spectrum of $\Tab$ is given by a sequence $(\srev_n(\al,\ao)^2)_{n\in\N}$ of algebraically simple eigenvalues such that
\[
\sup_{n\in\N} \, \Im \left(\srev_n(\al,\ao)^2\right) < 0.
\]
Moreover, any sequence of normalized eigenfunctions corresponding to these eigenvalues forms a Riesz basis.
\end{proposition}

\begin{figure} 
\includegraphics[width= \textwidth]{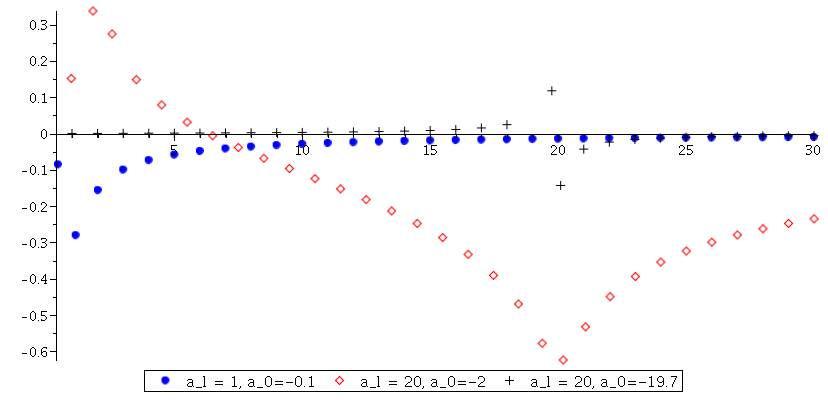}
\caption{$\lnab$ for $n \in \Ii 0 {30}$ and $\oo = \pi$.} \label{fig-30eigenvalues}
\end{figure}

\begin{proof}
As in the proof of Proposition \ref{prop-srev} we can see that for any $R > 0$ there exists $C_R \geq 0$ such that if $\l^2 \in \C^*$ is an eigenvalue of $\Tab$ we have
\begin{equation} \label{eq-ImRe-srev-bis}
\abs{\Re \srev} \leq R \quad \implies \quad \abs {\Im \srev} \leq C_R.
\end{equation}
The operator $\Tab$ depends analytically on the parameters $\al$ and $\ao$, and we know that when $\al = \ao = 0$ the eigenvalues $n^2\n^2$ for $n \in \N$ are algebraically simple. With the restrictions given by \eqref{eq-ImRe-srev-bis} and Lemma \ref{lemma-cross-nuN}, we obtain as in Section \ref{sec-transverse} a sequence of maps $(\al,\ao) \mapsto \srev_n(\al,\ao)$ such that the eigenvalues of $\Tab$ are $\srev_n(\al,\ao)^2$ for $n \in \N$. 
Let $n \in \N^*$. We have 
\begin{equation} \label{taylor-lnab}
\lnab = n\nu - \frac i {n\pi} (\al + \ao ) + \g (\al+\ao)^2 + O \left( \abs \al^3 , \abs \ao^3 \right),
\end{equation}
with $\Re(\g) = \oo/(n\pi)^3 > 0$. As in the dissipative case, we obtain that for any $\al,\ao$ with $\al + \ao > 0$ these eigenvalues $\lnab^2$ are simple.
If moreover $\al$ and $\ao$ are small enough, the eigenvalue $\lnab^2$ is close to $(n\n)^2$ and away from the real axis uniformly in $n \in \N^*$ (the first two terms in \eqref{taylor-lnab} are also the first two terms of the asymptotic expansion for large $n$ and fixed $\al$ and $\ao$). It remains to check that we also have $\Im(\srev_0(\al,\ao)^2) < 0$. For small $\al,\ao$ we denote by $\f_{\al,\ao}(0)$ a normalized eigenvector corresponding to the eigenvalue $\srev_0(\al,\ao)^2$ and depending analytically on $\al$ and $\ao$. For all $\p \in H^1(0,\oo)$ we have 
\[
\innp{\f_{\al,\ao}'}{\p'}_{L^2(0,\oo)}  - i \al {\f_{\al,\ao}(\oo)} \bar {\p(\oo)} - i \ao {\f_{\al,\ao}(0)} \bar {\p(0)} = \loab^2 \innp{\f_{\al,\ao}}{\p}_{L^2(0,\oo)} 
\]
We apply this with $\p = \f_{\al,\ao}$, take the derivatives with respect to $\al$ and $\ao$ at point $(\al,\ao) = (0,0)$, and use the facts that $\f_{0,0}$ is constant and $\srev_{0}(0,0) = 0$. We obtain
\[
\nabla_{\al,\ao} \big(\srev_0^2 \big) = -\frac   i \oo \big( 1 , 1 \big). 
\]
This proves that $\Im\big(\loab^2\big)<0$ if $\al$ and $\ao$ are small enough with $\al + \ao > 0$.
The Riesz basis property relies as before on the fact that 
\[
\abs{\lnab - n\n} = O(n\inv).
\]
For this point we can follow what is done in Section \ref{sec-transverse} for the dissipative case.
\end{proof}

For $n \in \N$ and $\al,\ao \in \R$ with $\al + \ao > 0$ we consider a normalized eigenvector $\fab \in L^2(0,\oo)$ corresponding to the eigenvalue $\lnab^2$ of $\Tab$. We denote by $(\f_n^*(\al,\ao))_{n\in\N}$ the dual basis.

\begin{proposition} \label{prop-solution-nondiss}
Let $u_0 \in \Dom(\Haoal)$. Then the problem \eqref{schrodinger-nondiss} has a unique solution $u \in C^1 (\R,L^2(\O)) \cap C^0(\R,\Dom(\Haoal))$. Moreover if we write 
\[
u_0 = \sum_{n\in\N} u_{0,n} \otimes \fab
\]
where $u_{0,n} \in L^2(\R^{\dimm -1})$, then $u$ is given by 
\[
u(t) = \sum_{n\in\N} \left( e^{-it \left(-\D_x+\lnab^2 \right)} u_{0,n}\right) \otimes \fab.
\]
\end{proposition}

\begin{proof} 
\stepp
Assume that $u \in C^0(\R_+, \Dom(\Haoal)) \cap C^1(\R_+^*,L^2(\O))$ is a solution of \eqref{schrodinger-nondiss}. Let $t \in \R_+^*$. For all $n \in \N$ and almost all $x \in \R^{\dimm-1}$ we can define 
\[
u_n(t,x) = \innp{u(t;x,\cdot)}{\f^*_n(\al,\ao)}_{L^2(0,\oo)},
\]
so that in $L^2(\O)$ we have
\[
u(t) = \sum_{n\in\N} u_n(t) \otimes \fab.
\]
According to Proposition \ref{prop-dec-res} (which can be proved similarly in this context) we have $u_n(t) \in H^2(\R^{\dimm-1})$ for all $t \in \R_+^*$ and $n \in \N$, and for $s \in \R^*$ we have 
\begin{multline*}
i\frac {u(t+s)-u(t)}s - \Haoal u(t) \\ = \sum_{n\in\N} \left( i\frac {u_n(t+s)-u_n(t)}s - \left(-\D_x +\lnab^2 \right) u_n(t) \right) \otimes \fab.
\end{multline*}
Let $n\in\N$. According to Proposition \ref{prop-riesz} we have 
\begin{multline*}
\nr{i\frac {u_n(t+s)-u_n(t)}s - \left(-\D_x +\lnab^2 \right) u_n(t)}_{L^2(\R^{\dimm-1})}
\\ 
\lesssim \nr{i\frac {u(t+s)-u(t)}s - \Haoal u(t)}_{L^2(\O)}
\limt s 0 0.
\end{multline*}
This proves that $u_n$ is differentiable and for all $t > 0$
\[
i u_n'(t) = \left( - \D_x + \lnab^2 \right) u_n(t). 
\]
Then for all $t > 0$
\[
 u_n(t) = e^{-it\left( -\D_x + \lnab^2 \right)} u_n(0) =   e^{-it\left( -\D_x + \lnab^2 \right)} u_{0,n}.
\]

\stepp 
Conversely, let us prove that the function $u$ defined by the statement of the proposition is indeed a solution of \eqref{schrodinger-nondiss}. 
Let $t \in \R$. According to Proposition \ref{prop-dec-res}, $u_{0,n}$ and hence $e^{-it \left(-\D_x+\lnab^2 \right)} u_{0,n}$ belong to $H^2(\R^{\dimm-1})$ for all $n \in \N$. Therefore $u(t) \in \Dom(\Haoal)$. Then for all $s \in \R^*$ we have 
\begin{eqnarray*}
\lefteqn{ \sum_{n\in\N}\nr{\left( i\frac {e^{-is \left(-\D_x + \lnab^2 \right)} - 1}s -\left(- \D_x + \lnab^2\right) \right) \,  e^{-it(-\D_x + \lnab^2)} u_{0,n}}^2 }\\
&& \lesssim_t \sum_{n\in\N}\nr{ \frac 1 s \int_0^s \big( e^{-i\th(-\D_x + \lnab^2)} - 1 \big) \left(- \D_x + \lnab^2\right) u_{0,n} \, d\th}^2_{L^2(\R^{\dimm-1})}\\
&& \lesssim_t \sum_{n\in\N}\nr{\left(- \D_x + \lnab^2\right) u_{0,n} }^2_{L^2(\R^{\dimm-1})}\\
&& \lesssim_t \nr{ \Haoal u_0}_{L^2(\O)}^2
\end{eqnarray*}
This series of functions converges uniformly in $s$ so we can take the limit, which proves that for any $t \in \R$
\[
\nr{i\frac {u(t+s)-u(t)}s -\Haoal u(t)}^2 \limt s 0 0.
\]
This proves that $u$ is differentiable and $i u'(t) + \Haoal u(t) = 0$, so $u$ is indeed a solution of \eqref{schrodinger-nondiss}.
\end{proof}

Now we can prove Theorem \ref{th-energy-decay-nondiss}:

\begin{proof} [Proof of Theorem \ref{th-energy-decay-nondiss}]
According to Proposition \ref{prop-solution-nondiss} we have existence and uniqueness for the solution $u$ of the problem \ref{schrodinger-nondiss}. Then with Proposition \ref{prop-riesz} we have 
\begin{align*}
\nr{u(t)}_{L^2(\O)}^2 
& \lesssim  \sum_{n \in \N} \nr{e^{-it \left(-\D_x + \lnab^2 \right)} u_{0,n}}_{L^2(\R^{\dimm-1})}^2 \\
& \lesssim  \sum_{n \in \N} e^{t \Im \left(\lnab^2\right)}\nr{ u_{0,n}}_{L^2(\R^{\dimm-1})}^2 .
\end{align*}
Proposition \ref{prop-Im-lab} gives $\g_{\al,\ao} > 0$ such that 
\begin{align*}
\nr{u(t)}_{L^2(\O)}^2 \lesssim  e^{-t \g_{\al,\ao}} \sum_{n \in \N}  \nr{ u_{0,n}}_{L^2(\R^{\dimm-1})}^2  \lesssim e^{-t \g_{\al,\ao}} \nr{u_0}_{L^2(\O)}^2,
\end{align*}
which concludes the proof.
\end{proof}

In the end of this section we show that the smallness assumption on $\abs{\al}+ \abs {\ao}$ is necessary in Theorem \ref{th-energy-decay-nondiss}. More precisely, if $\abs{\al}$ and $\abs{\ao}$ are too large, then the transverse operator $\Tab$ has eigenvalues with positive imaginary parts and hence the solution of the Schr\"odinger equation grows exponentially.

\begin{proposition} \label{prop-imvp-up}
Let $\al, \ao \in \R$ with $\al + \ao > 0$ and $\al \ao < 0$. Let $n \in \N$. If $s > 0$ is large enough, we have $\Im\big(\lnsab^2\big) > 0$.
\end{proposition}

\begin{proof}
We know that the curves $s \mapsto \lnsab$ for $n \in \N$  are defined for all $s \in \R$ and remain bounded. Moreover we have chosen the square root $\lnsab$ of $\lnsab^2$ which has a non-negative real part, so the imaginary parts of $\lnsab$ and $\lnsab^2$ have the same signs. Assume that $\al \ao < 0$, and let $n \in \N^*$ be fixed. We have
\begin{equation} \label{dev-lnab-large}
\frac {(\lnsab+s \al)(\lnsab+s\ao)}{(\lnsab-s\al)(\lnsab-s\ao)} = 1 + \frac {2\lnsab}s \left( \frac 1 {\al} + \frac 1 {\ao} \right) + \bigo s {+\infty} \big( s^{-2} \big). 
\end{equation}
Since $\Re (\lnsab) > n\n$ and $\Im (\lnsab)$ is bounded, this quantity is of norm less than 1 when $s > 0$ is large enough, so 
\[
 e^{-2\oo \Im(\lnsab)} = \abs{e^{2i\oo \lnsab }} < 1,
\]
and hence $\Im(\lnsab) > 0$. When $n = 0$, the same holds if we can prove that $\losab$ does not go to 0 for large $s$. Indeed, in this case the only possibility to have 
\[
e^{2i\oo \losab} \to 1
\]
is that $\losab$ goes to $\n$, and then $\Re(\losab)$ is bounded by below by a positive constant, and we can conclude as before. So assume by contradiction that $\losab$ goes to 0 as $s$ goes to $+\infty$. Then we have 
\[
e^{2i\oo \losab} = 1 + 2i\oo \losab + \bigo {s} {+\infty} \big( \abs{\losab)}^2 \big),
\]
which gives a contradiction with \eqref{dev-lnab-large}, where the rest $O(s^{-2})$ has to be replaced by $O\big( \l_0^2 s^{-2}\big)$. This concludes the proof.
\end{proof}

\begin{remark} \label{rem-anti-diss}
We remark from \eqref{taylor-lnab} (see also Figure \ref{fig-30eigenvalues}) that given $\al,\ao \in \R$ such that $\al +\ao > 0$ we always have $\Im (\lnab) < 0$ if $n$ is large enough. By duality, this means that there is always an eigenvalue with positive imaginary part when $\al + \ao < 0$, and hence the norm of the solution of \eqref{schrodinger-nondiss} is always exponentially increasing in this case.
\end{remark}

\bibliographystyle{alpha}
\bibliography{bibliotex}

\vspace{1cm}

% 
% \begin{center}
% 
% \begin{minipage}{0.8 \linewidth}
% 
% % \noindent  {Julien Royer}
% 
% \noindent  {\sc Institut de math\'ematiques de Toulouse}
% 
% \noindent  {\sc 118, route de Narbonne}
% 
% \noindent  {\sc 31062 Toulouse C\'edex 9}
% 
% \noindent  {\sc France}
% 
%  \begin{verbatim}
% julien.royer@math.univ-toulouse.fr
% \end{verbatim}
% \end{minipage}
% \end{center}

\end{document}